\documentclass[runningheads]{llncs}

\usepackage{microtype}
\usepackage{amsmath,amssymb}

\let\doendproof\endproof
\renewcommand\endproof{~\hfill\qed\doendproof}

\usepackage{mathtools}
\usepackage[nodollardollar, error]{onlyamsmath} 
\usepackage[inline]{enumitem}
\usepackage{mathtools}
\usepackage{tikz}
\usepackage{hyperref}

\newcommand{\cl}{\lvert\varphi\rvert}
\newcommand{\abs}[1]{\left\lvert#1\right\rvert}

\newcommand{\p}{\mathtt{p}}

\newcommand{\cupdot}{\mathbin{\mathaccent\cdot\cup}}
\newcommand{\set}[1]{\{#1\}}

\newcommand{\valuation}{\mathcal{V}}

\newcommand{\rLTL}{\mathrm{rLTL}}
\newcommand{\LTL}{\mathrm{LTL}}
\newcommand{\valueset}{\mathbb{B}_4}
\newcommand{\win}[1]{\mathrm{Win}(#1)}
\newcommand{\game}{\mathcal{G}}
\newcommand{\arena}{\mathcal{A}}
\newcommand{\reach}[1]{\mathrm{Reach}_1(#1)}
\newcommand{\pre}{\mathrm{pre}}

\newcommand{\enfMore}[1]{\mathrm{E}_{\geq #1}}
\newcommand{\enfEMore}[1]{\mathrm{E}_{> #1}}
\newcommand{\enfExact}[1]{\mathrm{E}_{= #1}}

\newcommand{\smrMore}[1]{\mathrm{V}_{> #1}}
\newcommand{\smrLess}[1]{\mathrm{V}_{< #1}}
\newcommand{\smrEMore}[1]{\mathrm{V}_{\geq #1}}
\newcommand{\smrExact}[1]{\mathrm{V}_{= #1}}

\newcommand{\summaryMax}{\mathrm{smry}}
\newcommand{\playsummary}[3]{\mathrm{smry}(#1,#2,#3)}
\newcommand{\stratsummary}[2]{\mathrm{smry}(#1,#2)}
\newcommand{\summaries}{\mathcal{S}}
\newcommand{\smry}{{s}}
\newcommand{\smrydash}{{s}'}

\newcommand{\blank}{\bot}
\newcommand{\blanks}{\blank,\ldots,\blank}

\newcommand{\lex}{\mathrm{lex}}
\newcommand{\gelex}{\ge_{\lex}}
\newcommand{\lelex}{\le_{\lex}}
\newcommand{\glex}{>_{\lex}}
\newcommand{\llex}{<_{\lex}}

\newcommand{\leftshift}[1]{\mathrm{lft}(#1)}

\newcommand{\bmfree}[2]{\Pi(\sigma,\p)}
\newcommand{\evade}[1]{\mathrm{Ev}(#1)}

\newcommand{\qnew}{v_{\mathrm{new}}}
\newcommand{\strong}{S}
\newcommand{\weak}{W}

\newcommand{\new}[1]{#1}

\colorlet{darkgreen}{green!80!black}
\colorlet{darkred}{red!80!black}
\usetikzlibrary{arrows, automata, shapes}
\tikzset{auto, >= stealth}
\tikzset{every edge/.append style={thick, shorten >= 1pt}}
\tikzset{initial/.style={draw, thick, <-, shorten <=1pt}}
\tikzset{player0/.style = {draw, thick, shape=circle, minimum size=8mm}}
\tikzset{player1/.style = {draw, thick, shape=rectangle, minimum size=8mm}}

\tikzset{robust/.style={line width=.16ex,line join=round}}

\let\Box\relax
\DeclareMathOperator{\Box}{%
	\text{%
		\tikz[baseline]{%
    			\draw[robust] (0ex,-.1ex) -- (0ex, 1.4ex) -- (1.5ex, 1.4ex) -- (1.5ex, -.1ex) -- cycle;%
		}%
	}%
}

\DeclareMathOperator{\Boxdot}{%
	\text{%
		\tikz[baseline]{%
    			\draw[robust] (0ex, -.1ex) -- (0ex, 1.4ex) -- (1.5ex, 1.4ex) -- (1.5ex, -.1ex) -- cycle;%
	    		\fill (.75ex, .65ex) circle (.15ex);%
    		}%
	}%
}

\let\Diamond\relax
\DeclareMathOperator{\Diamond}{%
	\text{%
		\tikz[baseline]{%
			\draw[robust] (0ex,.6ex) -- (.95ex, 1.55ex) -- (1.9ex, .6ex) -- (.95ex, -.35ex) -- cycle;%
		}%
	}%
}

\DeclareMathOperator{\Diamonddot}{%
	\text{%
		\tikz[baseline]{%
			\draw[robust] (0ex,.6ex) -- (.95ex, 1.55ex) -- (1.9ex, .6ex) -- (.95ex, -.35ex) -- cycle;%
			\fill (.95ex, .6ex) circle (.15ex);%
		}%
	}%
}

\DeclareMathOperator{\Xdot}{%
	\text{%
		\tikz[baseline]{%
    			\draw[robust] (.75ex, .65ex) circle (.75ex);%
	    		\fill (.75ex, .65ex) circle (.15ex);%
    		}%
	}%
}

\DeclareMathOperator{\Udot}{%
	\text{%
		\tikz[baseline]{%
			\node[inner sep=0pt, anchor=base, font=\bfseries] {U};%
			\fill (.1ex, .9ex) circle (.15ex);%
		}%
	}%
}

\DeclareMathOperator{\Rdot}{%
	\text{%
		\tikz[baseline]{%
			\node[inner sep=0pt, anchor=base, font=\bfseries] {R};%
			\fill (-.025ex, 1.175ex) circle (.15ex);%
		}%
	}%
}

\begin{document}
\title{Robustness-by-Construction Synthesis:\\ Adapting to the Environment at Runtime}
\titlerunning{Robustness-by-Construction Synthesis}
\author{Satya Prakash Nayak\inst{1}\and Daniel Neider\inst{2}\and Martin Zimmermann\inst{3}}
%
\authorrunning{S.\ P.\ Nayak et al.}
%
\institute{Max Planck Institute for Software Systems,
Kaiserslautern, Germany \\
\email{sanayak@mpi-sws.org} \and
Safety and Explainability of Learning Systems Group \\
Carl von Ossietzky Universität Oldenburg, Oldenburg, Germany \\
\email{daniel.neider@uni-oldenburg.de} \and
Aalborg University, Aalborg, Denmark\\
\email{mzi@cs.aau.dk}}

\maketitle              %
\begin{abstract}
While most of the current synthesis algorithms only focus on correctness-by-construction, ensuring robustness has remained a challenge. Hence, in this paper, we address the robust-by-construction synthesis problem by considering the specifications to be expressed by a robust version of Linear Temporal Logic~($\LTL$), called robust $\LTL$~($\rLTL$). rLTL has a many-valued semantics to capture different degrees of satisfaction of a specification, i.e., satisfaction is a quantitative notion. 

We argue that the current algorithms for $\rLTL$ synthesis do not compute optimal strategies in a non-antagonistic setting. So, a natural question is whether there is a way of satisfying the specification ``better'' if the environment is indeed not antagonistic. We address this question by developing two new notions of strategies. The first notion is that of adaptive strategies, which, in response to the opponent’s non-antagonistic moves, maximize the degree of satisfaction. The idea is to monitor non-optimal moves of the opponent at runtime using multiple parity automata and adaptively change the system strategy to ensure optimality. The second notion is that of strongly adaptive strategies, which is a further refinement of the first notion. These strategies also maximize the opportunities for the opponent to make non-optimal moves.
We show that computing such strategies for $\rLTL$ specifications is not harder than the standard synthesis problem, e.g., computing strategies with $\LTL$ specifications, and takes doubly-exponential time.
\end{abstract}

\section{Introduction}
\label{section1}

Formal methods have focused on the paradigm of correctness-by-construction, i.e., ensuring that systems are guaranteed to meet their design specifications.
While correctness is necessary, it has widely been acknowledged that this property alone is insufficient for a good design when a reactive system interacts with an ever-changing, uncontrolled environment.
To illustrate this point, consider a typical correctness specification $\varphi \Rightarrow \psi$ of a reactive system, where $\varphi$ is an environment assumption and $\psi$ the system's desired guarantee. 
Thus, if the environment violates $\varphi$, the entire implication becomes vacuously true, regardless of whether the system satisfies $\psi$.
In other words, if the assumption about the environment is violated, the system may behave arbitrarily.
This behavior is clearly undesirable as modeling any reasonably complex environment accurately and exhaustively is exceptionally challenging, if not impossible.

The example above shows that reactive systems must not only be correct but should also be \emph{robust} to unexpected environment behavior.
The notion of robustness we use in this paper is inspired by concepts from control theory~\cite{DBLP:journals/tecs/MajumdarRT13,DBLP:conf/hybrid/SamuelMSN20,DBLP:conf/cdc/SamuelMSN20,DBLP:journals/tac/TabuadaCRM14} and requires that deviations from the environment assumptions result in at most proportional violations of the system guarantee.
More precisely, ``minor'' violations of the environment assumption should only cause ``minor'' violations of the system guarantee, while ``major'' violations of the environment assumption allow for ``major`` violations of the system guarantee.

To capture different degrees of violation (or satisfaction) of a specification, we rely on a many-valued extension of Linear Temporal Logic~(LTL)~\cite{LTL}, named \emph{robust Linear Temporal Logic~(rLTL)}, which has recently been introduced by Tabuada and Neider~\cite{DBLP:conf/csl/TabuadaN16}.
The basic idea of this logic can best be illustrated by considering the prototypical environment assumption $\varphi \coloneqq \Box p$ (``always~$p$''), which demands that the environment ensures that an atomic proposition~$p$ holds at every step during its interaction with the system.
Clearly, $\varphi$ is violated even if $p$ does not hold at a single step, which is a ``minor'' violation. However, the classical Boolean semantics of LTL cannot distinguish between this case and the case where $p$ does not hold at any position, which is a ``major'' violation.
To distinguish these (and more) degrees of violations, $\rLTL$ adopts a five-valued semantics with truth values $\mathbb B_4 = \{ 1111, 0111, 0011, 0001, 0000\}$. The set~$\mathbb B_4$ is ordered according to $1111 > 0111 > 0011 > 0001 > 0000$, where $1111 $ is interpreted as $\mathit{true}$ and all other values as increasing shades of $\mathit{false}$.
In case of the formula $\varphi$, for instance, the interpretation of these five truth values is as follows: $\varphi$ evaluates to $1111$ if the environment ensures $p$ at every step of the interaction, $\varphi$ evaluates to $0111$ if $p$ holds almost always, $\varphi$ evaluates to $0011$ if $p$ holds infinitely often, $\varphi$ evaluates to $0001$ if $p$ holds at least once, and $\varphi$ evaluates to $0000$ if $p$ never holds.
The semantics of $\rLTL$ is then set up so that $\varphi \Rightarrow \psi$ evaluates to $1111$ if any violation of the environment assumption $\varphi$ causes at most a proportional violation of the system guarantee $\psi$ (i.e., if $\varphi$ evaluates to truth value $b \in \mathbb B_4$, then $\psi$ must evaluate to a truth value $b' \geq b$).

Here, we are interested in the synthesis problem for $\rLTL$ specifications. As usual, we model such a synthesis problem as an infinite-duration two-player game. 
Since we study $\rLTL$ synthesis, we consider games with $\rLTL$ winning conditions, so-called $\rLTL$ games.

$\rLTL$ games with a Boolean notion of winning strategy for the system player have already been studied by Tabuada and Neider~\cite{DBLP:conf/csl/TabuadaN16}. In their setting, the objective for the system player is as follows: given a truth value $b \in \mathbb B_4$, he must react to the actions of the environment player in such a way that the specification is satisfied with a value of at least $b$.
As for $\omega$-regular games, a winning strategy for the system player can immediately be implemented in hardware or software.
This implementation then results in a reactive system that is guaranteed to satisfy the given specification with at least a given truth value $b \in \mathbb B_4$, regardless of how the environment acts.

While $\rLTL$ games provide an elegant approach to robustness-by-construction synthesis, the Boolean notion of winning strategies that Tabuada and Neider adopt has a substantial drawback: it does not incentivize the system player to satisfy the specification with a value better than $b$, even if the environment player allows this.
Of course, one can (and should) statically search for the largest $b \in \mathbb B_4$ such that the system player can win the game.
However, this traditional worst-case view does not account for many practical situations where the environment is not antagonistic, e.g., in the presence of intermittent disturbances or noise, or when the environment cannot be modeled entirely~\cite{DBLP:conf/cdc/DallalNT16,DBLP:conf/hybrid/EhlersT14,DBLP:conf/mfcs/NeiderT020,DBLP:conf/csl/NeiderW018,DBLP:conf/hybrid/TopcuOLM12}.
In such situations, the system player should exploit the environment's ``bad'' moves, i.e., actions that permit the system player to achieve a value greater than $b$, and adapt its strategy at runtime.

We present two novel synthesis algorithms for $\rLTL$ specifications that ensure that the resulting systems are \emph{robust by construction} (in addition to being correct by construction).
These are based on two refined non-Boolean notions of winning strategies for $\rLTL$ games which both optimize the satisfaction of the specification.

The first notion, named \emph{adaptive strategies}, uses automata-based runtime verification techniques~\cite{runtime_verif} to monitor plays, detect bad moves of the environment, and adapt the actions of the system player to optimize the satisfaction of the winning condition.
The second notion, named \emph{strongly adaptive strategies}, is an extension of the first one that, in addition to being adaptive, also seeks to maximize the opportunity for the environment player to make bad moves.
We show that both types of strategies can be computed using methods from automata theory and result in effective synthesis algorithms for reactive systems that are robust by construction and adapt to the environment at runtime.

After recapitulating $\rLTL$ in Section~\ref{section2}, we introduce adaptive strategies in Section~\ref{adaptive strategy} and show that one can compute such strategies in $\rLTL$ games in doubly-exponential time by reducing the problem to solving parity games~\cite{calude}.
In Section~\ref{sec:strongly-adaptive}, we then turn to strongly adaptive strategies.
It turns out that this type of strategy does not always exist, which we demonstrate through an example. 
Nevertheless, we give a doubly-exponential time algorithm that decides whether a strongly adaptive strategy exists, and, if this is the case, computes one. Our algorithm is based on reductions to a series of parity and obliging games~\cite{Obliging}.
As the $\LTL$ synthesis problem is 2EXPTIME-complete~\cite{LTL_synthesis}, which is a special case of the problems we consider here, computing both types of adaptive strategies is 2EXPTIME-complete as well.
Furthermore, the size of the (strongly) adaptive strategies our algorithms compute is at most doubly exponential, matching the corresponding lower bound for $\LTL$ games, demonstrating that this bound is tight.
Thus, our results show that adaptive robust-by-construction synthesis is asymptotically not harder than classical $\LTL$ synthesis.

All proofs omitted can be found in the appendix.


\paragraph*{Related Work.}

Robustness in reactive synthesis has been addressed in various forms.
A prominent example is work by Bloem et al.~\cite{DBLP:journals/acta/BloemCGHHJKK14}, which considers the synthesis of robust reactive systems from GR(1)-specifications.
In subsequent work, Bloem et al.~\cite{DBLP:journals/corr/BloemEJK14} have surveyed a large body of work on robustness in reactive synthesis and distilled three general categories:
\begin{enumerate*}[label={(\roman*)}]
	\item ``fulfill the guarantee as often as possible even if the environment assumption is violated'',
	\item ``if it is impossible to fulfill the guarantee, try to fulfill it whenever possible'' and
	\item ``help the environment to fulfill the assumption if possible''.
\end{enumerate*}
Prototypical examples include the work by Topcu et al.~\cite{DBLP:conf/hybrid/TopcuOLM12}, Ehlers and Topcu~\cite{DBLP:conf/hybrid/EhlersT14}, Chatterjee and Henzinger~\cite{DBLP:conf/tacas/ChatterjeeH07}, Chatterjee et al.~\cite{Obliging}, and Bloem et al.~\cite{DBLP:conf/cav/BloemCHJ09}. 

However, our notion of adaptive strategies is more closely related to the notion of subgame perfect equilibrium~\cite{subgame}. A strategy profile is a subgame perfect equilibrium if it represents a Nash equilibrium of every subgame of the original game, i.e., the strategies are not only required to be optimal for the initial vertex but for every possible initial history of the game. Subgame perfect equilibria have been well studied in the context of graph games with LTL objectives~\cite{subgame,subgame_ltl}. Our results on adaptive strategies can be used to extend this concept to games with rLTL objectives. In particular, a pair of adaptive strategies for both players in such games forms a subgame perfect equilibrium and vice versa. Moreover, a subgame perfect equilibrium in such games is more fine-grained than subgame perfect equilibria in games with LTL objectives due to the many-valued semantics of rLTL.
On the other hand, our notion of strongly adaptive strategies is more general than subgame perfect equilibria. 

Also, the work of Almagor and Kupferman~\cite{DBLP:conf/cav/AlmagorK20} is very similar to our notion of adaptive strategies. They introduced the notion of good-enough synthesis that is considered over a multi-valued semantics where the goal is to compute a strategy that achieves the highest possible satisfaction value.
While some of the methods mentioned above do adapt to non-antagonistic behavior of the environment, we are not aware of any approach that would additionally optimize for the opportunities of the environment to act non-antagonistically, as our notion of strongly adaptive strategies does.

Quantitative objectives in graph-based games (and their combination with qualitative ones) have a rich history.
Among the most prominent examples are mean-payoff parity games~\cite{DBLP:conf/lics/ChatterjeeHJ05} and energy parity games~\cite{DBLP:journals/tcs/ChatterjeeD12}.
The former type of game combines a parity winning condition (as the canonical representation for $\omega$-regular properties) with a real-valued payout whose mean is to be maximized, while the latter type seeks to satisfy an $\omega$-regular winning condition with the quantitative requirement that the level of energy during a play must remain positive.
However, to the best of our knowledge, research in this field has focused on worst-case analyses with antagonistic environments.

Our notion of (strongly) adaptive strategies relies on central concepts introduced in the logic $\rLTL$~\cite{DBLP:conf/csl/TabuadaN16}, a robust, many-valued extension of LTL.
One of $\rLTL$'s key features is its syntactic similarity to LTL, which allows for a seamless and transparent transition from specifications expressed in LTL to specifications expressed in $\rLTL$.
Moreover, it is worth mentioning that rLTL has spawned numerous follow-up works, including $\rLTL$ model checking~\cite{DBLP:conf/hybrid/AnevlavisNPT19,DBLP:conf/cdc/AnevlavisPNT18,DBLP:journals/tocl/AnevlavisPNT22}, $\rLTL$ runtime monitoring~\cite{DBLP:conf/hybrid/MascleNSTW020}, and robust extensions of prompt LTL and Linear Dynamic Logic~\cite{NEIDER2021104810}, as well as CTL~\cite{NRZ22}.

Finally, let us highlight that preliminary results on adaptive strategies have been presented as a poster at the 24th ACM International Conference on Hybrid Systems: Computation and Control~\cite{DBLP:conf/hybrid/NayakN021}.

\section{Preliminaries}\label{section2}
In this section, we describe the syntax and semantics of Robust $\LTL$ and how it is different from classical $\LTL$. Moreover, we discuss some important results on $\rLTL$ and introduce games with $\rLTL$ specifications.

\paragraph{Robust Linear Temporal Logic.} 
We assume that the reader is familiar with Linear Temporal Logic~\cite{LTL}. We fix a finite non-empty set $\mathcal{P}$ of atomic propositions. The syntax of $\rLTL$ is similar to that of $\LTL$ with the only difference being the use of dotted temporal operators in order to distinguish them from $\LTL$ operators. More precisely, $\rLTL$ formulas are inductively defined as follows:

\begin{itemize}
	\item each $p\in \mathcal{P}$ is an $\rLTL$ formula, and
	\item if $\varphi$ and $\psi$ are $\rLTL$ formulas, so are $\neg \varphi$, $\varphi \vee \psi$, $\varphi \wedge \psi$, $\varphi \Rightarrow \psi$, $\Xdot \varphi$ (``next''), $\Boxdot \varphi$ (``always''), $\Diamonddot \varphi$ (``eventually''), $\varphi\Rdot \psi$ (``release'') and $\varphi\Udot \psi$ (``until'').
\end{itemize}

As already discussed, $\rLTL$ uses the set~$\valueset = \{1111,0111,0011,0001,0000\}$ of truth values, which are ordered as follows:
\[1111>0111>0011>0001>0000.\]
Intuitively, $1111$ corresponds to ``true'', and the other four values correspond to different degrees of ``false''.

The $\rLTL$ semantics is a mapping $\valuation$, called \emph{valuation}, that maps an infinite word $\alpha \in (2^{\mathcal{P}})^{\omega}$ and an $\rLTL$ formula $\varphi$ to an element of $\valueset$. 
Before we define the semantics, we need to introduce some useful notation.
Let $\alpha = \alpha_0\alpha_1\cdots \in (2^{\mathcal{P}})^{\omega}$ be an infinite word. For $i\in \mathbb{N}$, let $\alpha_{i\ldots} = \alpha_i\alpha_{i+1}\cdots$ be the (infinite) suffix of $\alpha$ starting at position~$i$. Also, for $1\leq k \leq 4$, we let $V_k(\alpha, \varphi)$ denote the $k$-th entry of  $\valuation(\alpha, \varphi)$, i.e., $\valuation(\alpha, \varphi) = V_1(\alpha, \varphi) V_2(\alpha, \varphi) V_3(\alpha, \varphi) V_4(\alpha, \varphi)$.
Now, $V$ is defined inductively as follows, where the semantics of Boolean connectives relies on \textit{da Costa algebras}~\cite{da_costa}:
\begin{align*}
	\valuation(\alpha, p) & = \begin{cases} 0000 & \text{if $p\not \in \alpha_0$} \\ 1111 & \text{if $p \in \alpha_0$} \end{cases} &
	\valuation(\alpha, \lnot \varphi) & = \begin{cases} 0000 & \text{if $\valuation(\alpha, \varphi) = 1111$} \\ 1111 & \text{otherwise} \end{cases} \\
	\valuation(\alpha, \varphi \lor \psi) & = \max{\bigl\{ \valuation(\alpha, \varphi), \valuation(\alpha, \psi) \bigr\}} &
	\valuation(\alpha, \varphi \Rightarrow \psi) & = \begin{cases} 1111 & \text{if $\valuation(\alpha, \varphi)\leq \valuation(\alpha, \psi)$} \\ \valuation(\alpha, \psi) & \text{otherwise} \end{cases} \\
	\valuation(\alpha, \varphi \land \psi) & = \min{\bigl\{ \valuation(\alpha, \varphi), \valuation(\alpha, \psi)\bigr\}} &
	\valuation(\alpha,\Xdot \varphi) & = \valuation(\alpha_{1\ldots}, \varphi) \\
	\valuation(\alpha,\Boxdot \varphi) & = \parbox{0pt}{$\displaystyle \left(\inf_{i\geq 0} V_1(\alpha_{i\ldots}, \varphi), \adjustlimits \sup_{j\geq 0} \inf_{i\geq j} V_2(\alpha_{i\ldots}, \varphi), \adjustlimits \inf_{j\geq 0} \sup_{i\geq j} V_3(\alpha_{i\ldots}, \varphi), \sup_{i\geq 0} V_4(\alpha_{i\ldots}, \varphi)\right)$} \\
	\valuation(\alpha,\Diamonddot \varphi) & = \parbox{0pt}{$\displaystyle \left(\sup_{i\geq 0} V_1(\alpha_{i\ldots}, \varphi), \sup_{i\geq 0} V_2(\alpha_{i\ldots}, \varphi), \sup_{i\geq 0} V_3(\alpha_{i\ldots}, \varphi), \sup_{i\geq 0} V_4(\alpha_{i\ldots}, \varphi)\right)$}
\end{align*}
The semantics for the temporal operators $\Udot$ and $\Rdot$ can be generalized similarly.
We refer the reader to Tabuada and Neider~\cite{DBLP:conf/csl/TabuadaN16} for more details.

\begin{example}
We can see that for the formula $\Boxdot p$, the valuation $\valuation(\alpha,\Boxdot p)$ can be expressed in terms of the $\LTL$ valuation function~$W$ by 
\[\valuation(\alpha,\Boxdot p) = W(\alpha,\Box p) W(\alpha, \Diamond \Box p) W(\alpha,\Box \Diamond p) W(\alpha,\Diamond p).\]
This evaluates to different values in $\valueset$ distinguishing various degrees of violations as seen in Section~\ref{section1}. 
\end{example}

\begin{example}
Now let us see how the $\rLTL$ semantics for a specification of the form $\varphi\Rightarrow\psi$ captures robustness. Consider an instance where the environment assumption $\varphi$ is $\Boxdot p$ and the system guarantee $\psi$ is $\Boxdot q$ and assume the specification $\Boxdot p \Rightarrow \Boxdot q$ evaluates to $1111$ for some infinite word. Let us see how the system behaves in response to various degrees of violation of the environment assumption.
\begin{itemize}
	\item If $p$ holds at all positions, then $\Boxdot p$ evaluates to $1111$. Hence, by the semantics of implication, $\Boxdot q$ also evaluates to $1111$, which means $q$ holds at all positions. Therefore, the desired behavior of the system is retained when the environment assumption holds with no violation.
	
	\item If $p$ holds eventually always but not always (a minor violation of $\Boxdot p$), then $\Boxdot p$ evaluates to $0111$. Hence, $\Boxdot q$ evaluates to $0111$ or higher, meaning that $q$ also needs to hold eventually always.
	
	\item Similarly, if $p$ holds at infinitely (finitely) many positions, then $q$ needs to hold at infinitely (finitely) many positions.
\end{itemize}
Hence, the semantics of $\Boxdot p \Rightarrow \Boxdot q$ captures the robustness property as desired. Furthermore, if $\Boxdot p \Rightarrow \Boxdot q$ evaluates to $b<1111$, then $\Boxdot p$ evaluates to a higher value than $b$, whereas $\Boxdot q$ evaluates to $b$. So, the desired system guarantee is not satisfied. However, the value of $\Boxdot p \Rightarrow \Boxdot q$ still describes which weakened guarantee follows from the environment assumption.

\end{example}

\paragraph{From rLTL to B\"uchi Automata.}
Given an $\LTL$ formula $\varphi$, a generalized B\"uchi automaton (see \cite{Buchi-Parity} for a definition) with $O(2^{\cl})$ states and $O(\cl)$ accepting sets can be constructed that recognizes the infinite words satisfying $\varphi$~\cite{LTL_Buchi},  where $\cl$ denotes the number of subformulas of $\varphi$. Using a similar method, Tabuada and Neider obtained the following result.

\begin{theorem}[\cite{DBLP:conf/csl/TabuadaN16}]\label{thm:rLTL-Buchi}
	Given an $\rLTL$ formula $\varphi$ and a set of truth values $B\subseteq \valueset$, one can construct a generalized B\"uchi automaton $\mathcal{A}$ with $2^{O(\cl)}$ states and $O(\cl)$ accepting sets that recognizes the infinite words on which the value of $\varphi$ belongs to $B$, i.e., $L(\mathcal{A}) = \{w\in (2^{\mathcal{P}})^{\omega}\mid \valuation(\alpha,\varphi) \in B\}$.
\end{theorem}

\paragraph{rLTL Games.}

\label{section3}

We consider infinite-duration two-player games over finite graphs with $\rLTL$ specifications. Here, we assume basic familiarity with games on graphs. 
Formally, an $\rLTL$ game $\game = (\arena, \varphi)$ consists of 
	\begin{enumerate*}[label=(\roman*)]
		\item a finite, directed, labelled arena $\arena = (V,E,\lambda)$ with $V = V_0 \cupdot V_1$, an edge relation~$E \subseteq V \times V$, and a labelling function $\lambda \colon V \rightarrow 2^{\mathcal{P}}$, and
		\item an $\rLTL$ formula $\varphi$ over $\mathcal{P}$.
	\end{enumerate*}
The game is played by two players, Player~$0$ and Player~$1$, who construct a \textit{play} $\rho = v_0v_1\cdots \in V^\omega$ by moving a token along the edges of the arena. A~play $\rho = v_0v_1\cdots $ induces an infinite word $\lambda(\rho) = \lambda(v_0)\lambda(v_1)\cdots \in (2^{\mathcal{P}})^{\omega}$, and the \textit{value} of the play, denoted by $\valuation(\rho)$, is the value of the formula $\varphi$ on $\lambda(\rho)$. Player~$0$'s objective is to maximize this value, while Player~$1$'s objective is to minimize it. 

\paragraph*{Strategies.}
A \textit{play prefix} is a finite, nonempty path $\p\in V^*$ in the arena. 
Then, a \textit{strategy} for Player~$i$, $i\in \{0,1\}$, is a function $\sigma\colon V^*V_i \rightarrow V$ mapping each play prefix~$\p$ ending in a vertex in $V_i$ to one of its successors.
Intuitively, a strategy prescribes Player~$i$'s next move depending on the play prefix constructed so far.

A strategy $\sigma$ is \textit{memoryless} if it only depends on the last vertex, i.e., for every prefix~$\p$ ending in vertex~$v$, it holds that $\sigma(\p) = \sigma(v)$. Moreover, we say a strategy has \textit{memory size} $m$ if there exists a finite state machine with output with $m$ states computing the strategy (see Grädel et al.~\cite{Automata_Book} for more details).

Next we define the plays that are consistent with a given strategy for Player~$i$.
Typically, this means that the token is placed at some initial vertex and then, whenever a vertex of Player~$i$ is reached, then Player~$i$ uses the move prescribed by the strategy for the current play prefix to extend this prefix.
Note that the strategy does not have control over the initial placement of the token.

Here we will use a more general notion, inspired by previous work in optimal strategies for Muller games~\cite{FJ} and in subgame perfect equilibria in graph games~\cite{subgame,subgame_ltl}: the initial prefix over which the strategy does not have control over might be longer than just the initial vertex. 
This means strategies are also applicable to prefixes that where not constructed according to the strategy. 
However, crucially, the strategy still gets access to that prefix and therefore can base its decisions on the prefix it had no control over. 
This generality will turn out to be useful both when defining adaptive strategies and when combining strategies to obtain adaptive strategies.

Formally, for a play prefix~$\p=v_0v_1\cdots v_n$ and a strategy $\sigma$ for Player~$i$, a play~$\rho$ is a $(\sigma, \p)$-play if $\rho =\p v_{n+1}v_{n+2}\cdots$ with $v_{k+1} = \sigma(v_0v_1\cdots v_k)$ for all $v_k\in V_i$ with $k\geq n$. Note that the prefix~$\p$ is arbitrary here, i.e., it might not have been constructed following the strategy~$\sigma$. 
Moreover, a $(\sigma,\p)$-play prefix $\p\p'$ is a prefix of a $(\sigma,\p)$-play.
We say that a play~$\rho$ starting in some vertex~$v$ is consistent with $\sigma$, if it is a $(\sigma, v)$-play (which is the classical notion of consistency).
Finally, a play prefix~$\p$ is consistent with $\sigma$ if it is the prefix of some play that is consistent with $\sigma$.

In the paper introducing $\rLTL$~\cite{DBLP:conf/csl/TabuadaN16}, Tabuada and Neider gave a doubly-exponential time algorithm that solves the classical $\rLTL$ synthesis problem, which is equivalent to solving the following problem.
\begin{problem}
	Given an $\rLTL$ game $\game$, an initial vertex~$v_0$ and a truth value~$b\in \valueset$, compute a strategy $\sigma$ (if one exists at all) for Player~$0$ such that every $(\sigma,v_0)$-play has value at least $b$.
\end{problem}

Note that Tabuada and Neider were interested in strategies for Player~$0$ that enforce the value~$b$ from $v_0$, i.e., strategies such that every consistent play starting in the given initial vertex has at least value~$b$.
In contrast, we will compute strategies that are improvements in two dimensions: 
\begin{enumerate*}[label=(\roman*)]
\item they enforce the optimal value rather than a given one, and
\item they do so from every possible play prefix, even if they did not have control over the prefix.
\end{enumerate*}

\section{Adaptive Strategies}\label{adaptive strategy}
In this section, we start by presenting a motivating example, a game in which classical strategies for Player~$0$ are not necessarily optimal (in an intuitive sense). We then formalize this intuition by introducing adaptive strategies and give a doubly-exponential time algorithm to compute such strategies. 

\paragraph*{Motivating Example.}\label{sec: Motivating adaptive}
Consider the arena given in \autoref{fig:motivate adaptive}
(where Player~$0$'s vertices are shown as circles and Player~$1$'s vertices are shown as squares) with the $\rLTL$ specification~$\varphi =\Boxdot p$.

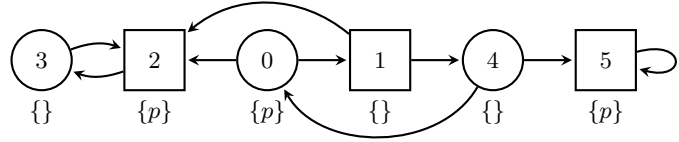
\begin{figure}[t!]
	\centering
	\begin{tikzpicture}
		\node[player0, label={below:$\{ p \}$}] (0) at (0, 0) {$0$};
		\node[player1, label={below:$\{  \}$}] (1) at (1.5, 0) {$1$};
		\node[player0, label={below:$\{  \}$}] (4) at (3, 0) {$4$};
		\node[player1, label={below:$\{ p \}$}] (5) at (4.5, 0) {$5$};
		\node[player1, label={below:$\{ p \}$}] (2) at (-1.5, 0) {$2$};
		\node[player0, label={below:$\{  \}$}] (3) at (-3, 0) {$3$};

		\path[->] (0) edge (1) edge (2);
		\path[->] (1) edge (4) edge[bend right =40] (2);
		\path[->] (4) edge (5) edge[bend left=60] (0);
		\path[->] (5) edge[loop right] ();
		\path[->] (2) edge[bend left=20] (3);
		\path[->] (3) edge[bend left=20] (2);
		
	\end{tikzpicture}
	\caption{First motivating example for adaptive strategies}
	\label{fig:motivate adaptive}
\end{figure}

Suppose the token is initially placed at vertex~$0$. Considering Player~$1$ plays optimally, the token would eventually reach vertex~$2$, from which the best possible scenario for Player~$0$ is to enforce a play where $p$ holds at infinitely many positions. As the classical problem only considers the worst-case analysis, a classical strategy for Player~$0$ is to try to visit vertex~$2$ infinitely often. That can be done by moving the token along one of the following edges every time the token reaches Player~$0$'s vertices: $\{0\rightarrow 1; 3\rightarrow 2; 4\rightarrow 0\}$. Note that the move $4\rightarrow 0$ is irrelevant in this worst-case analysis, as vertex~$4$ is never reached if Player~$1$ plays optimally. 

Suppose Player~$1$ makes a bad move by moving along $1\rightarrow 4$. Then, Player~$0$ can force the play to eventually just stay at vertex~$5$, and hence, $p$ holds almost always. However, the above classical strategy for Player~$0$ moves the play back to vertex~$0$, from which $p$ might not hold almost always. Therefore, a better strategy for Player~$0$ is to move along $4\rightarrow 5$ if the token reaches vertex~$4$ to get a play where $p$ holds almost always; otherwise, enforce a play where $p$ holds at infinitely many positions as earlier by moving along $0\rightarrow 1$ and then $3\rightarrow 2$ repeatedly. 

In the worst case, i.e., if Player~$1$ does not make a bad move by reaching vertex~$4$, both strategies yield value~$0011$.
However, if Player~$1$ does make a bad move by reaching vertex~$4$, the second strategy achieves value~$0111$ on some plays, while the second one does not.
So, in the worst case analysis, both strategies are equally good, but if we assume that Player~$1$ is not necessarily antagonistic, then the second strategy is better as it is able to exploit the bad move by Player~$1$.
We call such a strategy \textit{adaptive} as it adapts its moves to achieve the best possible outcome after each bad move of the opponent. We will formalize this shortly.

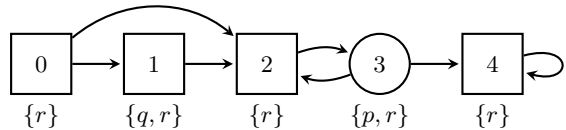
\begin{figure}[b!]
	\centering
	\begin{tikzpicture}
		\node[player1, label={below:$\{ r\}$}] (0) at (0, 0) {$0$};
		\node[player1, label={below:$\{ q,r \}$}] (1) at (1.5, 0) {$1$};
		\node[player1, label={below:$\{ r\}$}] (2) at (3, 0) {$2$};
		\node[player0, label={below:$\{ p,r \}$}] (3) at (4.5, 0) {$3$};
		\node[player1, label={below:$\{ r \}$}] (4) at (6,0) {$4$};

		\path[->] (0) edge (1) edge[bend left=40] (2);
		\path[->] (1) edge (2);
		\path[->] (2) edge[bend left=20] (3);
		\path[->] (3) edge[bend left=20] (2) edge(4);
		\path[->] (4) edge[loop right] ();
		
	\end{tikzpicture}
	\caption{Second motivating example for adaptive strategies}
	\label{fig:motivate adaptive2}
\end{figure}

To further illustrate the notion of adaptive strategies, consider another game with the arena shown in \autoref{fig:motivate adaptive2} and with $\rLTL$ specification $\varphi' = (\Xdot \neg q \Rightarrow \Boxdot p)\wedge (\Xdot q \Rightarrow \Boxdot r)$.
In this example, Player~$1$ has only two strategies starting from vertex~$0$: one moving the token along $0\rightarrow 1$ and one moving the token along $0\rightarrow 2$.

The best truth value Player~$0$ can enforce in this game is $0011$. 
This is because Player~$1$ can move along the edge~$0\rightarrow 2$, which satisfies the second implication with value~$1111$ (as $q$ does not occur), but also satisfies the premise of the first implication with value~$1111$. 
Hence, the value of the whole formula is the value of the subformula~$\Boxdot p$.
The best value Player~$0$ can achieve for it is indeed $0011$ by looping between vertices~$3$ and $2$.
His only other choice, i.e., to move to $4$ eventually, only results in the value~$0001$.

However, if Player~$1$ does not take the edge~$0\rightarrow 2$ but instead moves to vertex~$2$ via vertex~$1$, Player~$0$ can gain from this \emph{bad move} by instead moving to vertex~$4$.
In that case, the formula is satisfied with truth value~$1111$. 
Thus, a strategy that adapts to the bad move by the opponent can achieve a better value than one that does not, if they do make a bad move.

\subsection{Definitions}
Recall that a $(\sigma,\p)$-play for a strategy~$\sigma$ for Player~$i$ and a play prefix~$\p$ (not necessarily consistent with $\sigma$) is an extension of $\p$ by $\sigma$, i.e., Player~$i$ uses the strategy~$\sigma$ to extend the play prefix~$p$ \new{they} had not control over, while still taking the prefix~$\p$ into account when making the decisions.
We say that a strategy $\sigma$ for Player~$0$ \textit{enforces} a truth value of $b$ from a play prefix~$\p$, if we have $\valuation(\rho) \geq b$ for every $(\sigma, \p)$-play~$\rho$. 
Similarly, we say a strategy $\tau$ for Player~$1$ enforces a truth value of $b$ from a play prefix~$\p$, if we have $\valuation(\rho) \leq b$ for every $(\tau, \p)$-play~$\rho$. 
This conforms to our intuition that Player~$0$ tries to maximize the truth value while Player~$1$ tries to minimize it.
Moreover, we say Player~$i$ can enforce a value $b$ from some prefix~$\p$ if \new{they have} a strategy that enforces $b$ from $\p$.

\begin{remark}
Let $\p$ be a play prefix. If Player~$0$ can enforce value~$b_0$ from $\p$ and Player~$1$ can enforce value~$b_1$ from $\p$ then $b_0 \le  b_1$	.
\end{remark}

For example, consider the game given in \autoref{fig:motivate adaptive} with the $\rLTL$ specification~$\Boxdot p$. Using the analysis given in Section~\ref{sec: Motivating adaptive}, we can see that Player~$0$ can enforce $0111$ and $0011$ from prefixes $014$ and $012$, respectively, by moving the token along $\{0\rightarrow 1; 4\rightarrow 5; 3\rightarrow 2\}$. It is easy to check that these are the best values Player~$0$ can enforce from those prefixes as Player~$1$ can enforce the same values from these prefixes.

We are interested in a strategy that enforces the best possible value from each play prefix. This is formalized as follows.

\begin{definition}[Adaptive Strategies]\label{def:adaptive}
	In an $\rLTL$ game, a strategy $\sigma_0$ for Player~$0$ is \textit{adaptive} if from every play prefix~$\p$, no strategy for Player~$0$ enforces a better truth value than $\sigma_0$, 
	that is, if some strategy $\sigma$ for Player~$0$ enforces a truth value of $b$ from $\p$, then $\sigma_0$ also enforces the value $b$ from $\p$.
\end{definition} 

Note that $\p$ is not required to be consistent with $\sigma_0$ in the above definition, i.e., an adaptive strategy achieves the best possible outcome from every possible play prefix (even for those it had no control over when they are constructed).
Also, let us mention that a dual notion can be defined for Player~$1$. 

As mentioned earlier, one can notice that the concept of adaptive strategies is similar to the concept of subgame perfect equilibrium~\cite{subgame}. Furthermore, adaptive strategies can be used to extend the notion of subgame perfect equilibrium to rLTL games as in the following remark.

\begin{remark}
Given an rLTL game, a pair of adaptive strategies for both players forms a subgame perfect equilibrium for the game, and vice versa.
\end{remark}

Here, we prefer the notion of adaptive strategies over the notion of subgame perfect equilibria, as we focus on strategies for Player~$0$ and generally disregard strategies of Player~$1$.

\subsection{Computing Adaptive Strategies}\label{solving rLTL games}
Now, to synthesize an adaptive strategy, we need to monitor the bad moves of the opponent at runtime by keeping track of the best value that can be enforced from the current play prefix. 
To do that, using the idea of automata-based runtime verification~\cite{runtime_verif}, we construct multiple parity automata to monitor the bad moves of the opponent and then we synthesize adaptive strategies by using a reduction to parity games (see \cite{Automata_Book} for definitions). 

Given an $\rLTL$ game $\game = (\arena, \varphi)$ with $\arena = (V,E,\lambda)$, we proceed as follows:
\begin{enumerate}
	\item We construct  generalized (non-deterministic) B\"uchi automata $\mathcal{A}^b$ such that $L(\mathcal{A}^b) =\{w \in  (2^\mathcal{P})^{\omega} \mid \valuation(w, \varphi) \geq b\}$ for all $b\in \valueset$. 
	
	\item We determinize each $\mathcal{A}^b$ to obtain a deterministic parity automaton~$\mathcal{C}^b$ with the same language.
	
	\item For each $b$, we construct a parity game $\game^b$ by taking the product of the arena~$\arena$ and the parity automaton~$\mathcal{C}^b$.
	
	\item We solve the above parity games~$\game^b$~\cite{calude}, yielding, for each truth value~$b$, a finite-state  winning strategy for the original game~$\game$ with value~$b$ (if one exists). 
	
	\item We combine all these winning strategies for Player~$0$ computed in the last step to obtain an adaptive strategy $\sigma$ for Player~$0$.
	
\end{enumerate}

Let us now explain each step in more detail.

\paragraph{Step 1.} We construct the generalized non-deterministic B\"uchi automata $\mathcal{A}^b$ such that $L(\mathcal{A}^b) =\{w \in  (2^\mathcal{P})^{\omega} \mid \valuation(w, \varphi) \geq b\}$ for all $b\in \valueset$. By \autoref{thm:rLTL-Buchi}, the automaton $\mathcal{A}^b$ has \new{$2^{O(\cl)}$} states and \new{$O(\cl)$} accepting sets.

\paragraph{Step 2.} We determinize each $\mathcal{A}^b$ to get a deterministic parity automaton $\mathcal{C}^b = (Q^b, 2^\mathcal{P},q_0^b,\delta^b,\Omega^b)$ with \new{$2^{2^{O(\cl)}}$} states and \new{$2^{O(\cl)}$} colors~\cite{Buchi-Parity}.

\paragraph{Step 3.} We construct the (unlabelled) product arena~$\arena^b = (V^b,E^b)$ of the arena~$\arena = (V,E,\lambda)$ and the parity automaton $\mathcal{C}^b$ such that $V^b = V \times Q^b$, $V_i^b = V_i \times Q^b$ for $i \in \set{0,1}$, and
\[((v,q),(v',q'))\in E^b \text{ if and only if } (v,v')\in E \text{ and } \delta^b(q, \lambda(v)) = q'.\]
The function $\bar{\Omega}^b$ assigns colors to the vertices such that $\bar{\Omega}^b(v,q) = \Omega^b(q)$.
The desired parity games are the $\game^b = (\arena^b, \bar{\Omega}^b)$ with $b \in \valueset$.

It is easy to verify that Player~$0$ wins a play $\rho' = (v_0, q^b_0)(v_1, q^b_1)\cdots$ in $\game^b$ if and only if the value of the play $\rho = v_0v_1\cdots$ in $\game$ is at least $b$. Furthermore, given a path $\rho = v_0v_1\cdots v_k$ in $\arena$, there is a unique path of the form $\rho' = (v_0, q_0^b)(v_1, q_1^b)\cdots (v_k,q_k^b)$ in $\arena^b$, that is when $q^b_{i+1} = \delta^b(q_i^b,v_i)$ for all $0\leq i \leq k-1$. 

Since winning a play in $\game^b$ is equivalent to the corresponding play in $\game$ satisfying $\varphi$ with truth value $b$ or greater, we can characterize the enforcement of $b$ in $\game$ by the winning region of Player~$0$ in $\game^b$, i.e., the set of vertices from which Player~$0$ has a winning strategy.
This can easily be shown by simulating a winning strategy from $(v,q)$ to extend the play prefix~$\p $ and vice versa. 

\begin{remark}
\label{remark_enforcementcharac}	
Fix a play prefix~$\p$ in the rLTL game~$\game$, and let $(v,q^b)$ be the last vertex of the corresponding play in the parity game~$\game^b$ for some $b$.
Then, Player~$0$ can enforce~$b$ from $\p$ if and only if $(v,q^b)$ is in his winning region of $\game^b$.
\end{remark}

\paragraph{Step 4.} We solve the resulting parity games $\game^b$ and determine the winning regions~$\win{\game^b}$ of Player~$0$ and  uniform memoryless winning strategies~$\sigma^{b}$ for Player~$0$ that are winning from every vertex in the corresponding winning region. The parity games have \new{$n = \abs{V} \cdot 2^{2^{O(\cl)}}$} vertices and \new{$k = 2^{O(\cl)}$} colors. Since \new{$k < \lg(n)$}, these can be solved in time \new{$O(n^5) =  \abs{V}^5 \cdot 2^{2^{O(\cl)}}$}~\cite{calude}.

\paragraph{Step 5.} Consider the extended $\rLTL$ game $\game' = (\arena',\varphi)$, where $\arena' = (V', E', \lambda')$  with $V' = V \times Q^{0000} \times \cdots \times Q^{1111}$, \begin{multline*}
	E' = \bigl\{ \left((v_1,q_1^{0000}, \ldots, q_1^{1111}),(v_2,q_2^{0000}, \ldots, q_2^{1111})\right) \mid \\
(v_1,v_2)\in E \text{ and } \delta^{b}(q_1^b, \lambda(v)) = q_2^b \text{ for all }b\in \valueset \bigr\},
\end{multline*}
	and $\lambda'$ such that $\lambda'(v,q^{0000},\ldots,q^{1111}) = \lambda(v)$ for all $v\in V$ and $ q^b\in Q^b$.

It is easy to see that there is a one to one correspondence between the plays in both games $\game$ and $\game'$. Besides that, the $\rLTL$ specification is also the same in both games. Therefore, computing an adaptive strategy in the game $\game$ is equivalent to computing one in the game $\game'$. Now using the analysis given in Step~$3$, we have the following in the $\rLTL$ game $\game'$:
\begin{itemize}
	\item[\textbullet] A vertex~$v'$ is in $\enfMore{b} = \{(v,q^{0000},\ldots,q^{1111})\in V'\mid (v,q^b)\in \win{\game^{b}}\}$ if and only if Player~$0$ can enforce $b$ from every play prefix in $\game'$ ending in~$v'$.

	\item[\textbullet] Using these sets, we now define the set~$\enfExact{b}$ of vertices from which the maximum value Player~$0$ can enforce is $b$.
	Formally, this set is given by 
	\[
	\enfExact{b} = \begin{cases}
 	\enfMore{1111} &\text{if $b =1111$,} \\
 	\enfMore{b}\setminus \enfMore{b+1} &\text{if $b < 1111$},
 \end{cases}\]
where $b+1$ is the smallest value bigger than~$b<1111$. Note that the sets~$\enfExact{b}$ form a partition of the vertex set of $\game'$.
\end{itemize}

Furthermore, it is easy to see that if a play $\rho$ satisfies a parity objective then every play sharing a suffix with $\rho$ also satisfies the parity objective. Since the game $\game'$ is a product of parity games and since we have characterized the enforcement of truth values via the membership in the winning regions of the parity games (see Remark~\ref{remark_enforcementcharac}), the next remark follows.

\begin{remark}\label{remark: two prefix same ending}
	In the $\rLTL$ game $\game'$, for two play prefixes $\p_1,\p_2$ ending in the same vertex, the following holds: if a memoryless strategy~$\sigma$ for Player~$0$ enforces a truth value $b$ from $\p_1$, then it also enforces the value $b$ from $\p_2$.
\end{remark}

Then, we can see that if the token stays in $\enfExact{b}$ for some $b$, then Player~$0$ can simulate the strategy $\sigma^b$ for $\game^b$ to enforce the value $b$ in $\game'$. Therefore, we obtain a memoryless adaptive strategy $\sigma$ for Player~$0$ in the game $\game'$ as follows: for every vertex~$(v,q^{0000},\ldots,q^{1111})$ in $\enfExact{b}$, we define $\sigma(v,q^{0000},\ldots,q^{1111})$ to be the unique successor of $(v,q^{0000},\ldots,q^{1111})$ in $\game'$ that corresponds to the successor~$\sigma^b(v,q^b)$ of $(v,q^b)$ in $\game^b$.
Thus, $\sigma$ simulates the strategy $\sigma^b$ for the largest $b$ such that the value $b$ can be enforced (which is exactly what $\sigma^b$ does from such a prefix).
Hence, it is an adaptive strategy for Player~$0$ in $\game'$.

Finally, using the strategy $\sigma$, one can compute a corresponding strategy in the game $\game$ with memory $Q^{0000}\times Q^{0001} \times \cdots \times Q^{1111}$, which is used to simulate the positional strategy~$\sigma$. 
The resulting finite-state strategy is an adaptive strategy for Player~$0$ in $\game$.

Note that the adaptive strategy in $\game$ is of doubly-exponential size in $|V|$ and $\cl$.
This upper bound is tight,  since there is a doubly-exponential lower bound on the size of winning strategies strategies for $\LTL$ games~\cite{LTL_synthesis_memory}, which can be lifted to $\rLTL$ games.

Similarly, one can compute an adaptive strategy for Player~$1$. 
Hence, an adaptive strategy for both players in an $\rLTL$ game can be computed in time doubly-exponential in the size of the formula.  

\begin{theorem}
	Given an $\rLTL$ game, an adaptive strategy of a player can be computed in doubly-exponential time. Moreover, each player has an adaptive strategy with doubly-exponential memory size.
\end{theorem}

Note that adaptive strategies enforce the best possible value from the given prefix. 
This value can be obtained at runtime as follows: Given a play prefix~$\p$ ending in some vertex~$v$, let $(q^{0000},\ldots, q^{1111})$ be the state of the automaton implementing the adaptive finite-state strategy computed above is in after the prefix~$\p$. 
Note that this state has to be tracked to determine the next move the strategy prescribes at prefix~$\p$ (in case $v \in V_0$).
Then, there is a unique~$b$ such that $(v, q^{0000},\ldots, q^{1111}) \in \enfExact{b}$.
Then, the value currently enforced by the adaptive strategy is $b$, which, by construction, is the maximal one that can be enforced from $\p$.


\section{Strongly Adaptive Strategies}
\label{sec:strongly-adaptive}

In the previous section, we have argued the importance of adaptive strategies and proved that in every rLTL game both players have an adaptive strategy.
Intuitively, such a strategy exploits bad moves of the opponent to always enforce the best truth value possible after a given prefix. 
However, such a strategy does not necessarily seek out opportunities for the opponent to make bad moves. 
We argue that this property implies that some adaptive strategies are more desirable than others, which leads us to the notion of strongly adaptive strategies.

In this section, we define strongly adaptive strategies, which are based on a fine-grained analysis of the possibilities a strategy gives the opponent to make bad moves and the resulting outcomes of such bad moves.
We show that strongly adaptive strategies do not exist in every rLTL game.
This is in stark contrast to adaptive strategies, which always exists. 
Nevertheless, we give a doubly-exponential time algorithm that decides whether a strongly adaptive strategy exists and, if yes, computes one.

\subsection{Bad Moves}\label{sec:bad moves}
We already have used the notion of bad moves in Section~\ref{adaptive strategy} in an intuitive, but informal, way. 
Formally, we say a play~$\rho = v_0v_1 \cdots$ contains a bad move of Player~$i$ at position~$j>0$ if the player can enforce some value $b$ from the prefix~$v_0 \cdots v_{j-1}$ but can no longer enforce the value~$b$ from the prefix~$v_0\cdots v_{j}$.
Note that the position~$j$ is the target of the bad move. 
Moreover, note that moving from $v_0 \cdots v_{j-1}$ to $v_{j}$ can only be a bad move for Player~$i$ if it is Player~$i$'s turn at $v_{j-1}$.
Also, there must be some other edge from $v_{j-1}$ to a vertex~$v\not = v_{j}$ so that \new{they} can still enforce $b$ from $v_0 \cdots v_{j-1} v$. 

For the example given in \autoref{fig:motivate adaptive}, we know that Player~$1$ can enforce the value $0011$ from $01$ (by moving the token from $01$ to $2$). 
Suppose she moves the token from $01$ to $4$ instead. 
Then, Player~$0$ can enforce $0111$ by visiting vertex~$5$. 
Hence Player~$1$ can no longer enforce $0011$ from $014$. 
Therefore, the move from prefix~$01$ to vertex~$4$ made by Player~$1$ is bad. 

Note that if Player~$1$ makes a bad move from a play prefix~$\p$ to vertex~$v$, then the maximum value Player~$0$ can enforce from $\p v$ is strictly larger than the maximum value he can enforce from $\p$. 
Hence, if the maximum value Player~$0$ can enforce from a play prefix is $1111$, then Player~$1$ can not make any bad move from that prefix.
Moreover, assuming Player~$0$ does not make any bad move, the maximum value Player~$0$ can enforce from any play prefix can increase at most four times during a play. Thus, Player~$1$ can make at most four bad moves against an adaptive strategy, because such a strategy does not make any bad moves.

\begin{remark}
\label{remark:badmoves}
Let $\sigma$ be an adaptive strategy and $\p$ a play prefix (not necessarily consistent with $\sigma$). Then, every $(\sigma,\p)$-play~$\rho = v_0v_1 \cdots$ contains at most four bad moves of Player~$1$ after $\p$.
Also, if there is no bad move by Player~$1$ at positions~$j_0, j_0+1, \ldots, j_1$ in $\rho$, then $\sigma$ enforces the same truth values from every prefix of the form~$v_0 \cdots v_{j}$ with $j_0-1 \le j < j_1$.
\end{remark}

Our next example shows that an adaptive strategy does not actively seek out opportunities for the opponent to make bad moves, it just exploits those made.

\subsection{Motivating Example}\label{sec:motivate strongly}
Recall the example given in \autoref{fig:motivate adaptive}.
The strategy for Player~$0$ given by $\{\mbox{$0\rightarrow 1$}; 3\rightarrow 2; 4\rightarrow 5\}$ is adaptive: if Player~$1$ makes a bad move by moving from $1$ to $4$, then moving from $4$ to $5$ improves the value of the play to $0111$.
Such an improvement can only be enforced after the bad move. 

Another adaptive strategy for Player~$0$ is to move along $0\rightarrow 2$ directly in his first move and then move along $3\rightarrow 2$ every time.
Then, the token can never reach vertex~$1$. 
Hence, Player~$1$ can never make a bad move.
However, it also means that there can not be a play with value $0111$.
By contrast, if Player~$0$ moves along $0\rightarrow 1$, there is a chance of getting such plays (when Player~$1$ makes a bad move of $1\rightarrow 4$).
Therefore, using the earlier strategy of moving the token along $0\rightarrow 1$, Player~$0$ might be able to enforce $0111$ at some point, but he can never achieve the value~$0111$ when moving directly to vertex~$2$.

Similarly, in many games, a player may have two (or more) optimal choices to move the token from some prefix. In such situations, that player should compare the bad moves \new{their} opponent can make in both choices and determine the choice in which \new{they} can enforce the best value after a bad move has been made by the opponent. To capture this, we refine the notion of adaptive strategies by introducing strongly adaptive strategies, which are, in a sense to be formalized below, the best adaptive strategies.

\subsection{Definitions}

In this section, we introduce the necessary machinery to define strongly adaptive strategies for Player~$0$. 
Throughout this section, we are concerned with ranking adaptive strategies according to the number of bad moves they allow Player~$1$ to make, and on the effect these moves have. 
For the sake of conciseness, unless stated otherwise, from now on a bad move always refers to a bad move by Player~$1$.

We begin by introducing a ranking of plays and then lift this to strategies.
As the number of bad moves in one play is bounded by four, this results in at most five truth values that can be enforced from prefixes of the play, i.e., the one that is enforced before the first bad move, and the ones after each bad move.
If Player~$1$ makes less than four bad moves during a play, we use the symbol~$\blank \notin \valueset$ to signify~this.

We collect this information in a summary, a five-tuple~$(b_0, \ldots, b_k,\blanks) \in (\valueset \cup \set{\blank})^5$ such that $\blank \neq b_0 < b_1 < \cdots < b_k$.
The set of all summaries is denoted by $\summaries$. 

Fix an adaptive strategy~$\sigma$ for Player~$0$, a play prefix~$\p$ not necessarily consistent with $\sigma$, and a $(\sigma,\p)$-play~$\rho$, and let $0 \le k \le 4$ be the number of bad moves by Player~$1$ after $\p$.
Define $\p_0 = \p$ and let $\p_{j}$, for $1 \le j \le k$, be the prefix of $\rho$ ending at the position of the $j$-th bad move.
Due to Remark~\ref{remark:badmoves}, these prefixes contain information about all possible truth vales that are enforced by Player~$0$ from prefixes of $\rho$. 
We employ summaries to capture the values a given \emph{strategy}~$\sigma$ enforces from these prefixes. 
Formally, for $0 \le j \le k$, let $b_j$ be the maximal value that $\sigma$ enforces from $\p_j$.
As $\sigma$ is adaptive, these values are strictly increasing.
So, we can define the summary
$\playsummary{\sigma}{\p}{\rho} = (b_0, \ldots, b_k, \blanks)$.
Intuitively, the summary collects all information about which truth values the strategy~$\sigma$  enforces after each bad move has been made.
If there are less than four bad moves in $\rho$ after $\p$, then we fill the summary with $\blank$'s to obtain a vector of length five.

We will use such summaries to compare strategies.
To do so, we compare summaries in lexicographic order~$\lelex$ with $\blank$ being the smallest element. 
In other words, we prefer larger truth values of smaller ones and prefer the opportunity for a bad move over the impossibility of a bad move.

\begin{example}
\label{example:summaries}
Consider again the game in Figure~\ref{fig:motivate adaptive}. 
Let $\sigma_1$ be the memoryless Player~$0$ strategy always making the moves~$\set{0 \rightarrow 1; 3\rightarrow 2; 4\rightarrow 5 }$.
Then,
 \[\playsummary{\sigma_1}{0}{0145^\omega} = \playsummary{\sigma_1}{01}{0145^\omega} = (0011,0111,\blank,\blank,\blank)\]
and $\playsummary{\sigma_1}{014}{0145^\omega} = (0111,\blank,\blank,\blank,\blank)$ because the play~$0145^\omega$ does not contain a bad move of Player~$1$ after $014$.
In addition, $\playsummary{\sigma_1}{\p}{01(23)^\omega} = (0011,\blank,\blank,\blank,\blank)$ for every prefix~$\p$ of $01(23)^\omega$, as the play does not contain any bad move of Player~$1$.

Let $\sigma_2$ now be the memoryless Player~$0$ strategy given by $\set{0 \rightarrow 2; {3\rightarrow 2}; {4\rightarrow 5} }$.
Then, we have $\playsummary{\sigma_2}{\p}{0(23)^\omega} = (0011,\blank,\blank,\blank,\blank)$ for every prefix~$\p$ of  $0(23)^\omega$ because the play does not contain bad moves of Player~$1$.
\end{example}

We continue by listing some simple properties of summaries that are useful later on.
Consider the prefixes $0$, $01$, $014$ of $0145^\omega$ in Example~\ref{example:summaries}. The former two have the same summary~$\smry$, while the summary of the latter is obtained by shifting $\smry$ to the left. Note that moving from $01$ to $4$ is a bad move of Player~$1$, while moving from $0$ to $1$ is not.
By inspecting the definition of play summaries, it is clear that extending plays by bad moves corresponds to a left shift, while Remark~\ref{remark:badmoves} implies that the absence of bad moves keeps summaries stable.

To formalize this, we use the following notation:
for $\smry = (b_0, \ldots, b_k, \blanks) \in \summaries$ with $k>0$ let $\leftshift{\smry} = (b_1, \ldots, b_k, \blanks)  \in \summaries$, i.e., we shift $\smry$ to the left and fill the last entry with a $\blank$.
As entries in summaries are strictly increasing, we have $\leftshift{\smry} \glex \smry$ for every $\smry$ with at least two non-$\blank$ entries.

\begin{remark}
\label{remark:playsummaryevolution}
Let $\sigma$ be an adaptive strategy for Player~$0$, let $\p$ be a play prefix, and let $\rho = v_0 v_1 \cdots$ be a $(\sigma, \p)$-play. Further, let $n = |\p|$, i.e., $v_{n-1}$ is the last vertex of $\p$, and note that $\rho$ is also a $(\sigma, \p v_n)$-play.

If $\rho$ has a bad move at position~$n$, then $\playsummary{\sigma}{ \p v_n}{ \rho} = \leftshift{\playsummary{\sigma}{ \p}{ \rho}}$ (reflecting the fact that $\rho$ has one bad move less after $\p v_n$ than after $\p$), otherwise we have $\playsummary{\sigma}{ \p v_n}{ \rho} = \playsummary{\sigma}{ \p}{ \rho}$.
Note that we have kept $\sigma$ and $\rho$ fixed and just added a vertex to the prefix we consider.
\end{remark} 

As seen above, a bad move shifts the summary to the left.
The following remark shows a dual result, allowing us to determine the summary of a play prefix of length one from the summary of play prefix up to the first bad move. 
In Example~\ref{example:summaries}, note that the strategy~$\sigma_1$ (using the edges $\set{0 \rightarrow 1; 3\rightarrow 2; 4\rightarrow 5 }$) enforces value~$0011$ from $0$,  i.e., the first entry of $\playsummary{\sigma_1}{0}{0145^\omega}$ is $0011$.
The play~$0145^\omega$ has its first bad move of Player~$1$ at position~$2$, and the corresponding summary is~$\playsummary{\sigma_1}{014}{0145^\omega} = (0111,\blank,\blank,\blank,\blank)$.
Hence, $\playsummary{\sigma_1}{0}{0145^\omega}$ must be the ``concatenation''~$(0011,0111,\blank,\blank,\blank)$ of $0011$ and $(0111,\blank,\blank,\blank,\blank)$ (with the last $\blank$ removed).
In general, we have the following property.

\begin{remark}\label{remark:leftshiftSummary}
Let $\smry = (b_0,\ldots, b_k,\blanks) \in \summaries$ with $k>0$ and let $v$ be a vertex.
Let $\sigma$ be an adaptive strategy such that $b_0$ is the maximal value that $\sigma$ enforces from $v$
and let $\rho$ be a $(\sigma, v)$-play with at least one bad move, and let $\p$ be the prefix of $\rho$ ending at the position of the first bad move.
Then, $\playsummary{\sigma}{\p}{\rho} = \leftshift{\smry}$ if and only if $\playsummary{\sigma}{v}{\rho} = \smry$.
\end{remark}

Again, recall Example~\ref{example:summaries}, and consider the plays~$\rho_b = 0145^\omega$ (with a bad move by Player~$1$) and $\rho_n = 01(23)^\omega$ (without a bad move), which are both $(\sigma_1,0)$-plays.
We have $\playsummary{\sigma_1}{0}{\rho_b} = (0011,0111,\blank,\blank,\blank)$ and $\playsummary{\sigma_1}{0}{\rho_n} = (0011,\blank,\blank,\blank,\blank)$.
Disregarding the $\blank$'s the summary of $\rho_n$ can be seen as a strict prefix of the summary of $\rho_b$.
Note that $(0011,\blank,\blank,\blank,\blank) \llex (0011,0111,\blank,\blank,\blank)$. 

In general, fix a strategy~$\sigma$, a play prefix~$\p$, and a $(\sigma,\p)$-play~$\rho$ with $\playsummary{\sigma}{\p}{\rho} = (b_0, \ldots,b_k,\blanks)$. Then, for every $k' < k$ there is a $(\sigma, \p)$-play~$\rho'$ with $\playsummary{\sigma}{\p}{\rho'} = (b_0, \ldots,b_{k'},\blanks)$, i.e., any play where Player~$1$ stops making bad moves after the first $k'$ ones (recall that making bad moves is a choice).

To formalize this, we say that a summary~$(b_0, \ldots, b_k, \blanks)$ is a strict prefix of a summary~$(b_0', \ldots, b_{k'}', \blanks)$ if $k < k'$ and $b_j = b_j'$ for all $0\le j \le k$, i.e., we only consider non-$\blank$ entries.
Now, fix $(\sigma,\p)$-plays $\rho, \rho'$. 
We say that $\rho$ is $(\sigma, \p)$-covered by $\rho'$ if $\playsummary{\sigma}{\p}{\rho}$ is a strict prefix of $\playsummary{\sigma}{\p}{\rho'}$.
Also, we say that $\rho$ is a $(\sigma, \p)$-uncovered play if there is no $(\sigma,\p)$-play~$\rho'$ that covers it. 
When $\sigma$ and $\p$ are clear from context, we drop them and say that a play is uncovered.
In the example, $\rho_n$ is $(\sigma_1, 0)$-covered  by $\rho_b$, which is $(\sigma_1, 0)$-uncovered.

Now, we lift summaries from plays to strategies by defining~$\stratsummary{\sigma}{\p}$ as the lexicographical minimum over all $\playsummary{\sigma}{\p}{\rho}$ where $\rho$ ranges over $(\sigma, \p)$-uncovered plays. 
Note that if $\rho$ $(\sigma, \p)$-covers $\rho'$, then the summary of $\rho$ is a strict prefix of the summary of $\rho'$ and, therefore, strictly smaller.
Our definition of $\stratsummary{\sigma}{\p}$ discards such plays when computing the minimum, but the information is not lost as it appears as a prefix of a covering play.

In the running example, we have $\stratsummary{\sigma_1}{0} = (0011,0111,\blank,\blank,\blank)$ and $\stratsummary{\sigma_2}{0} = (0011,\blank,\blank,\blank,\blank)$.

\begin{remark}\label{remark:existenceofplaysummary}
Let $\sigma$ be an adaptive strategy for Player~$0$ and let $\p$ be a play prefix. If $\stratsummary{\sigma}{\p} = \smry$ for some $\smry\in \summaries$, then there exists a $(\sigma,\p)$-uncovered play $\rho$ such that $\playsummary{\sigma}{\p}{\rho} = \smry$.
\end{remark}

Finally, we are ready to formalize our intuitive notion of strongly adaptive strategies, i.e., adaptive strategies that seek out opportunities for the opponent to make bad moves. 
Recall that summaries record the possibility, and the effect, of Player~$1$ making bad moves. 
So, we intuitively say a strategy is strongly adaptive if it maximizes the summaries globally.

Recall that a strategy is adaptive if the value it enforces from any possible play prefix is as large as the value any other strategy enforces from that prefix.
Analogously, a strategy is strongly adaptive if its summary for every play prefix is as good as the summary from the play prefix for any other strategy. 

\begin{definition}
An adaptive strategy~$\sigma_0$ is strongly adaptive if $\stratsummary{\sigma_0}{\p}\gelex \stratsummary{\sigma}{\p}$ for every adaptive strategy~$\sigma$ and every play prefix~$\p$.	
\end{definition}

Before we start our proof, let us introduce one more useful bit of notation. 
For every play prefix~$\p$, let $\summaryMax(\p)$ denote the lexicographical maximum of $\stratsummary{\sigma}{\p}$ over all adaptive strategies $\sigma$ for Player~$0$ in the game $\game$, i.e.,
\[\summaryMax(\p) = \max_{\sigma} \stratsummary{\sigma}{\p},\]
where $\sigma$ ranges over all adaptive strategies for Player~$0$.

Note that every strongly adaptive strategy is adaptive by definition, and the first entry of $\summaryMax(\p)$ is equal to the maximal value that can be enforced from $\p$. However, as argued above, not every adaptive strategy is strongly adaptive, so in particular a strongly adaptive strategy for Player~$0$ never makes a bad move (of Player~$0$). 

\subsection{Existence of Strongly Adaptive Strategies}

While strongly adaptive strategies generalize adaptive strategies, there is a catch in the definition: The former may not always exist, whereas the latter always do. 
To show this, consider the graph given in \autoref{fig:no strongly ex} with initial vertex~$0$ and the formula~$\varphi = \Boxdot p$. It is clear that Player~$0$ can enforce $0011$ from every play prefix in $0(10)^*$ by eventually moving to vertex~$3$. And if at some point, Player~$1$ makes the bad move~$1\rightarrow 2$, then Player~$0$ enforces $0111$ as the token stays at vertex~$2$ forever. 
However, every adaptive strategy for Player~$0$ has to eventually visit vertex~$3$, unless Player~$1$ makes a bad move prior. 

Note that Player~$1$ can only make a bad move at vertex~$1$, so visiting $1$ once more when at vertex $0$ instead of moving to vertex~$3$ gives her another chance to make a bad move. 
So, to optimize the enforced value under one bad move, Player~$1$ should stay in the loop between $0$ and $1$ forever.
However, this is not the optimal behavior if no bad move occurs, as looping yields a value of $0000$, which is smaller than the value $0011$ that is achieved by eventually moving to $3$.

Formally, for $n\ge 0$, let $\sigma_n$ be the strategy such that $\sigma_n(0(10)^{n'}) = 1$ for all $n' < n$ and $\sigma_n(0(10)^{n'}) = 3$ for all $n' \ge n$, i.e., $\sigma_n$ gives Player~$1$ $n$ chances to make a bad move and then moves to $3$, thereby preventing her from making a bad move.
Note that each of the $\sigma_n$ is adaptive, but $\sigma_{n+1}$ gives Player~$1$ more opportunities to make a bad move than $\sigma_{n}$, namely for the prefix~$0(10)^{n}$.

Fix some $n$. 
There are only two $(\sigma_{n+1},0(10)^{n})$-plays, i.e., $0(10)^{n}10(34)^\omega$ (Player~$1$ does not make a bad move) and $0(10)^{n}12^\omega$ (Player~$1$ makes a bad move).
Then, $\playsummary{\sigma_{n+1}}{0(10)^{n}}{0(10)^{n}10(34)^\omega} = (0011,\blank,\blank,\blank,\blank)$ as well as $\playsummary{\sigma_{n+1}}{0(10)^{n}}{0(10)^{n}12^\omega} = (0011,0111,\blank,\blank,\blank)$.
Hence, we conclude $\stratsummary{\sigma_{n+1}}{0(10)^{n}} = (0011,0111,\blank,\blank,\blank)$, as the former is covered by the latter. 

Towards a contradiction assume there is a strongly adaptive strategy $\sigma$.
By definition, we have 
\begin{equation}
\label{eq}
	\stratsummary{\sigma}{0(10)^{n}} \gelex \stratsummary{\sigma_{n+1}}{0(10)^{n}}  = (0011,0111,\blank,\blank,\blank)
\end{equation} for every $n$.
As we have $\stratsummary{\sigma}{0(10)^{n}} \llex (1111,\blank,\blank,\blank,\blank)$ ($p$ does not hold at vertex~$0$), $\sigma$ must give Player~$1$ the chance to make at least one bad move after the prefix~$0(10)^{n}$.
So, we must have $\sigma(0(10)^{n}) = 1$, as Player~$1$ can only make a bad move at vertex $1$.

Thus, the play~$(01)^\omega$ (with value~$0000$) is a $(\sigma, 0(10)^{n})$-play for every $n$, i.e., $\sigma$ only enforces $0000$ from every such prefix.
Hence, the first entry of $\stratsummary{\sigma}{0(10)^{n}}$ is $0000$ for every $n$.
This contradicts Inequality~(\ref{eq}).
Therefore, $\sigma$ is not strongly adaptive, i.e., Player~$0$ does not have a strongly adaptive strategy in the game.

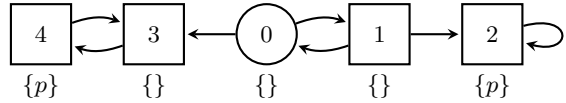
\begin{figure}[t!]
	\centering
	\begin{tikzpicture}
		\node[player0, label={below:$\{ \}$}] (0) at (0, 0) {$0$};
		\node[player1, label={below:$\{ \}$}] (1) at (1.5, 0) {$1$};
		\node[player1, label={below:$\{ p\}$}] (2) at (3, 0) {$2$};
		\node[player1, label={below:$\{\}$}] (3) at (-1.5,0) {$3$};
		\node[player1, label={below:$\{ p\}$}] (4) at (-3,0) {$4$};

		\path[->] (0) edge (3) edge[bend left=20] (1);
		\path[->] (1) edge (2) edge[bend left=20] (0);
		\path[->] (2) edge[loop right] ();			    
		\path[->] (3) edge[bend left=20] (4);
		\path[->] (4) edge[bend left=20] (3);
	\end{tikzpicture}
	\caption{An $\rLTL$ game with no strongly adaptive strategy}
	\label{fig:no strongly ex}
\end{figure}

As strongly adaptive strategies do not necessarily exist, we are interested in the following problem.
\begin{problem}\label{problem 2}
	Given an $\rLTL$ game, determine whether a strongly adaptive strategy for Player~$0$ exists and, if yes, compute one.
\end{problem}

\subsection{Computing Strongly Adaptive Strategies}
We solve Problem \ref{problem 2} for an $\rLTL$ game $\game = (\arena,\varphi)$ by constructing the parity games $\game^b$ for each $b$ and the extended game $\game' = (\arena',\varphi)$ as in the algorithm given in Section~\ref{solving rLTL games}.
Recall that Player~$0$ wins $\game^b$ if and only if he can enforce $b$ in $\game$ and that $\game'$ is the product of the $\game^b$.
As we have described in Step~$5$ of that algorithm, it is easy to see that solving Problem \ref{problem 2} for the game $\game$ is equivalent to solving the problem for game $\game'$. Hence, from now on, we only consider $\game'$ and show properties for the game $\game'$, which we can use later to compute a strongly adaptive strategy in $\game$. This strategy can then be transformed into a strongly adaptive strategy for $\game$.
%
%
%
In the following, it is often useful to focus on one truth value by equipping $\game'$ with the parity condition of $\mathcal C_b$ for some $b$: a vertex~$(v,q_{1111}, \ldots, q_{0000})$ has the color that $q_b$ has in $\mathcal C_b$. 
Thus, $\game'$ equipped with the parity condition of $\game^b$ is equivalent to $\game^b$.

To decide whether a strongly adaptive strategy exists, we proceed as follows: 
\begin{enumerate}
	\item We first give a characterization of the vertices~$v$ of $\game'$ with $\summaryMax(v)=\smry$ that only uses summaries that are larger than $\smry$. This allows us to compute $\summaryMax(v)$ for every vertex $v$ by induction over the summaries.
	\item Using the decomposition of $\game'$ into regions with the same summary and the characterization we construct a series of obliging games~\cite{Obliging}. In an obliging game, Player~$0$ has a strong winning condition that has to be satisfied on every play and a weak winning condition that must be satisfiable if Player~$1$ cooperates. In our case, the strong winning condition requires Player~$0$ to always enforce the best value that is currently possible and the weak condition requires Player~$1$ to have a chance to make a bad move (if the summary encodes that this is still possible), i.e., whenever possible, Player~$1$ is given the chance to make a bad move.
	\item Finally, if  Player~$1$ has in all obliging games a strategy satisfying both the strong and the weak condition, then these can be turned effectively into a strongly adaptive strategy, otherwise there is no such strategy.
\end{enumerate}

We first provide a useful lemma showing that a strategy in $\game'$ is strongly adaptive if and only if its summary is history independent, i.e., only depends on the last vertex.

\begin{lemma}\label{lem:strongly ad if and only if e(p) is summaryMax(v)}
A strategy $\sigma$ for Player~$0$ in $\game'$ is strongly adaptive if and only if for every play prefix~$\p$ ending in vertex~$v$, it holds that $\stratsummary{\sigma}{\p}) = \summaryMax(v)$.
\end{lemma}

As computed in Section~\ref{solving rLTL games}, let $\enfExact{b}$ be the set of vertices in $\game'$ from which the maximum value Player~$0$ can enforce is~$b$. 
Furthermore, for a summary $\smry \in \summaries$, let $\smrEMore{\smry}$ denote the set of vertices $v$ in $\game'$ for which $\summaryMax(v) \gelex \smry$. 
Let $\smrExact{\smry}$, $\smrMore{\smry}$, and $\smrLess{\smry}$ be defined similarly. 

\begin{remark}
Let $\smry = (b_0, \ldots, b_k,\blanks)$. Then, $\smrExact{\smry} \subseteq \enfExact{b_0}$.
\end{remark}

For a vertex set $F$, let $\pre(F)$ denote the set of vertices from which there is an edge to $F$.
Maybe surprisingly, we do not distinguish between vertices of Player~$0$ and Player~$1$, but we will only apply $\pre(F)$ when it is Player~$1$'s turn.

Next, we characterize the sets $\smrExact{\smry}$ in terms of the existence of strategies that witness summaries. 
The key aspects of this characterization is that it only refers to summaries~$\smrydash \glex \smry$, which will later allow us to compute these sets inductively. 

Given a strategy~$\sigma$ for Player~$0$ in $\game'$ and a play prefix~$\p$, let $\bmfree{\sigma}{\p}$ denote the set of $(\sigma,\p)$-plays that do not contain a bad move by Player~$1$ after $\p$.

\begin{definition}
Let $\sigma$ be a strategy for Player~$0$, $\p$ be a play prefix, and $\smry = (b_0, \ldots, b_k, \blanks)$ a summary. 
We say that $\sigma$ is an $\smry$-witness from $\p$ if and only if it satisfies the following three properties:

		\begin{description}
			\item[Enforcing] Every play in $\bmfree{\sigma}{\p}$ satisfies the parity condition of the game $\game^{b_0}$. Thus, a witness has to enforce~$b_0$ unless Player~$1$ makes a bad move.
			
			\item[Enabling] If $k \ge 1$, there exists a play in $\bmfree{\sigma}{\p}$ that visits $\pre\big( \smrExact{\leftshift{\smry}}\big)$. Thus, if there is the chance to reach a vertex where Player~$1$ can make a bad move, then a witness has to visit such a vertex. 
				Note that we require that the bad move leads to a vertex with summary~$\leftshift{\smry}$, which is the largest summary that can be guaranteed to be reached from $\p$ after a bad move. 
			
			\item[Evading] If $k \ge 1$, then let us define $\evade{\smry}$ to be the set of summaries $\smrydash = (b_0',\ldots, b_{k'}',\blanks)$ with $b_0' > b_0$, $\smrydash \llex \leftshift{\smry}$, and such that $\smrydash$ is not a strict prefix of $\leftshift{\smry}$. Then, no play in $\bmfree{\sigma}{\p}$ visits $\pre\big(\smrExact{\smrydash}\big)$ for any $\smrydash\in \evade{\smry}$.
			Thus, a witness can never reach a vertex where Player~$1$ can make a bad move to reach a summary that is worse than $\leftshift{\smry}$. 
		\end{description}
\end{definition}
Recall that $\leftshift{\smry} \glex \smry$ and that $\smrydash \in \evade{\smry}$ implies $\smrydash \glex \smry$.

Our next lemma shows that witnesses do indeed witness the summaries of vertices.

\begin{lemma}\label{lem:characterization}
	In the game $\game'$, for some summary $\smry$ and for some vertex~$v$, we have $v\in \smrExact{\smry}$ if and only if $v \not \in \smrMore{\smry}$ and there is an $\smry$-witness from $v$.
\end{lemma}

We now give a method to compute $\smrExact{\smry}$ for each summary~$\smry \in \summaries$ by induction from the largest to the smallest summary. 
Since truth values in a summary are strictly increasing, $(1111,\blanks)$ is the maximal summary.
We have $\smrExact{(1111,\blanks)} = \enfExact{1111}$, which we can compute using Tabuada and Neider's result for classical rLTL games (see Section~\ref{section3}).
For the inductive step, assume that for a summary~$\smry = (b_0, \ldots, b_k, \blanks)$ the sets $\smrExact{\smrydash}$ are already computed for every $\smrydash \glex \smry$.
The set $\smrExact{\smry}$ can then be computed using the following algorithm:

\begin{enumerate}
\item\label{item:algo1:1} If $ k = 0$, then return $ \enfMore{b_0}\setminus \smrMore{\smry} $. Here, $\enfMore{b_0}$ can again be computed using Tabuada and Neider's result for classical rLTL games.

\item\label{item:algo1:3} Now assume $k > 0$.
Let $\arena_{\smry}$ be the subgraph of $\arena'$ restricted to the vertex set 
\[ \enfMore{b_0}\setminus \Big(\smrMore{\smry} \cup \bigcup\nolimits_{\smrydash\in \evade{\smry}} \reach{\smrExact{\smrydash}}\Big),\]
where $\reach{F}$ denotes the set of vertices of $\arena'$ from which Player~$1$ can force the token to reach $F$. This set can be computed in linear time (in the number of edges of $\arena'$) using standard methods to solve reachability games (see \cite{Automata_Book} for more details). 
In the proof of correctness of the algorithm (see Lemma~\ref{lemma:algoforsummary}) we show that $\arena_{\smry}$ does not have any terminal vertices.

Also, the sets~$\smrExact{\smrydash}$ for $\smrydash \in \evade{\smry}$ are already computed, because  the summaries~$\smrydash \in \evade{\smry}$ are all greater than $\smry$.

\item\label{item:algo1:4} Let $\win{{\smry}}$ be the winning region for Player~$0$ in the parity game with arena~$\arena_{\smry}$ and coloring as in the game~$\game^{b_0}$.
Return the set of vertices in Player~$0$'s winning region~$\win{\smry}$ from which $\pre\big(\smrExact{\leftshift{\smry}}\big)$ is reachable in the subgraph of $\arena'$ restricted to $\win{{\smry}}$.
\end{enumerate} 

\begin{lemma}\label{lemma:algoforsummary}
The algorithm described above computes the sets~$\smrExact{\smry}$ for $\smry \in \summaries$.	
\end{lemma}

Now, we give a characterization of strongly adaptive strategies in terms of summary witnesses.

\begin{lemma}\label{lem:characterizationStronglyad}
In the game $\game'$, a strategy~$\sigma$ is strongly adaptive if and only if it is a $\summaryMax(\p)$-witness from every play prefix $\p$.
\end{lemma}

Now, we show how to decide whether a strategy satisfying the condition given in Lemma~$\ref{lem:characterizationStronglyad}$ exists, i.e., a strategy that is a $\summaryMax(\p)$-witness from every play prefix~$\p$.
Furthermore, if such a strategy exists, we compute one.
To do so, we present a reduction to another type of game, called \textit{obliging games}. So, before describing the details of the reduction, let us recapitulate the definitions and useful results on obliging games.

Obliging games are two-player games introduced by Chatterjee et al.~\cite{Obliging}. They have two winning conditions, $\strong$ and $\weak$, called strong and weak conditions. The objective of Player~$0$ is to ensure the strong winning condition while allowing Player~$1$ to cooperate with him to additionally fulfil the weak winning condition. Formally, a strategy $\sigma$ for Player~$0$ is \textit{uniformly gracious} if it satisfies the following:
\begin{itemize}
\item for every vertex~$v$, every $(\sigma,v)$-play is $\strong$-winning, and
\item for every play prefix $\p$ consistent with  $\sigma$, there is a $\weak$-winning $(\sigma,\p)$-play.
\end{itemize}
We are only interested in parity/B\"uchi obliging games (i.e., the strong condition is a parity condition, and the weak one is a B\"uchi condition). The next theorem follows directly from the results by Chatterjee et al.~\cite{Obliging}.

\begin{theorem}\label{thm: Obliging}
A parity/B\"uchi obliging game with $n$ vertices and a parity condition
with $k$ colors can be reduced to a parity game with $O(n)$ vertices and $O(k)$ colors. Moreover, if Player~$0$ has a uniformly gracious strategy in such an obliging game, he has a uniformly gracious strategy with a memory of size at most $O(k)$. 
\end{theorem}

Now, coming back to our problem, we define obliging games~$\game_{\smry}$ (for each $\smry\in \summaries$), which are subgames of $\game'$, such that a uniformly gracious strategy in $\game_{\smry}$ satisfies the properties of an $\smry$-witness locally. In particular, the games are defined in way such that the strong condition resembles the Enforcing property, the weak condition resembles the Enabling property, and the restricted vertex set ensures that the Evading property is satisfied.

\begin{definition}
Given a summary $\smry = (b_0,\ldots,b_k,\blanks)\in \summaries$, let $\game_{\smry} $ be the obliging game obtained from $\game'$ as follows:
\begin{itemize}
\item The set of vertices $V(\game_{\smry})$ is the set $\smrExact{\smry}\cup \{\qnew\}$, where $\qnew$ is a new vertex that does not belong to $V'$.
\item The set of edges $E(\game_{\smry})$ contains the following edges:
	\begin{itemize}
	\item The edges of the game $\game'$ restricted to the vertex set $\smrExact{\smry}$.
	\item All edges of the form $(v,\qnew)$ where $v$ is a terminal vertex in the game~$\game'$ restricted to $\smrExact{\smry}$.
	\item A self loop on $\qnew$.
	\end{itemize}
\item The strong condition $\strong_{\smry}$ is a min-parity condition such that the color of $\qnew$ is $0$ and the color of any other vertex is the same as in $\game^{b_0}$.
\item If $k=0$, then there is no weak condition, i.e., $\weak_{\smry}$ is a B\"uchi condition with $F = V(\game_{\smry})$. If $k>0$, then the weak condition $\weak_{\smry}$ is a B\"uchi condition with $F = \pre(\leftshift{\smry}) \cup \{\qnew\}$.
\end{itemize}
\end{definition}

The following lemma formalizes the connection between uniformly gracious strategies in the obliging games~$\game_\smry$ and strongly adaptive strategies in $\game'$.

\begin{lemma}\label{lemma:gracioustostronglyad}
There exists a strongly adaptive strategy in $\game'$ if and only if there exists a uniformly gracious strategy in every obliging game $\game_{\smry}$.
Given a uniformly gracious strategy with finite memory in each obliging game $\game_{\smry}$, one can effectively combine these into a strongly adaptive strategy with finite memory in $\game'$.
\end{lemma}

Since the game $\game'$ has doubly-exponential size, using \autoref{thm: Obliging}, the parity\slash{}Büchi obliging games $\game_{\smry}$ can be reduced to doubly-exponential-sized parity games. Once we computed a strongly adaptive strategy for $\game'$, it can then be reduced to a strongly adaptive strategy for the original game $\game$.

Moreover, note that strongly adaptive strategies also have doubly-exponential memory since the obliging games we constructed have doubly-exponential size.
By \autoref{thm: Obliging}, uniformly gracious strategies in such obliging games require memory of linear size, leading to the following result.

\begin{theorem}
Given an $\rLTL$ game, one can decide in doubly-exponential time whether Player~$0$ has a strongly adaptive strategy. If yes, one can compute one with doubly-exponential memory  in doubly-exponential time.
\end{theorem}

Note that by dualizing the definitions and the constructions, an analogous result for Player~$1$ can also be obtained.

\section{Conclusion}

We argued that in a reactive system, in addition to correctness, we also need to ensure robustness. 
To this end, we introduced adaptive strategies for $\rLTL$ games that satisfy the specification to a higher degree when the environment is not antagonistic. We also presented a stronger version of adaptive strategies that additionally maximizes the opportunities for the opponent to make bad choices. Finally, we showed that both adaptive and strongly adaptive strategies can be computed in doubly-exponential time. As we know that the classical $\LTL$ and $\rLTL$ synthesis algorithms also take doubly-exponential time, we conclude that adaptive and strongly adaptive strategies are not harder to compute.

\new{In future work, we aim to investigate even more general notions of adaption to the behavior of a not necessarily antagonistic environment. 
Possible approaches include observing the environment's behavior and trying to compare that to optimal strategies for the environment.}

 
\bibliographystyle{splncs04}
\bibliography{bib}

\newpage
\appendix

In this appendix, we present the proofs omitted in the main part. 

\section{Combining Adaptive Strategies}

Let us begin by introducing a  useful preliminary result:
in many of the constructions presented below, we need to combine several strategies into a new one while maintaining adaptiveness.
The following lemma will be useful to prove this.

\begin{lemma}\label{lem:follow adaptive is adaptive}
Let $\sigma$ be a strategy for Player~$0$ in $\game'$ such that for all play prefixes~$\p$ and all $(\sigma,\p)$-plays~$\rho$, the following are both satisfied: 
\begin{itemize}
\item $\rho$ does not contain a bad move of Player~$0$ after the play prefix~$\p$.
\item There is an adaptive strategy~$\sigma_\rho$ such that $\rho$ is a $(\sigma_\rho,\p')$-play for some prefix~$\p'$ of $\rho$.
\end{itemize}
Then, $\sigma$ is adaptive.
\end{lemma}
\begin{proof}
Towards a contradiction, suppose $\sigma$ is a strategy for Player~$0$ in $\game'$ satisfying the given properties, but is not adaptive.
Then, by definition, for some play prefix~$\p$, there exists some strategy $\sigma'$ that enforces some value~$b$ from $\p$, but $\sigma$ does not. 
That means there exists a $(\sigma,\p)$-play~$\rho$ which has a value strictly less than $b$. 
By the given properties, there exists an adaptive strategy $\sigma_\rho$ such that $\rho$ is a $(\sigma_\rho, \p')$-play for some prefix~$\p'$ of $\rho$. 
Since both $\p$ and $\p'$ are prefixes of $\rho$, one has to be the prefix of the other. 
If $\p'$ is a prefix of $\p$, then $\rho$ is also a $(\sigma_\rho,\p)$-play.
As $\sigma_\rho$ is adaptive, it enforces value~$b$ from $\p$, i.e., we have derived a contradiction to $\valuation(\rho) < b$. 

Now, assume that $\p$ is a prefix of $\p'$. 
Since $\sigma'$ enforces value $b$ from $\p$ and Player~$0$ does not make a bad move after prefix~$\p$ in the play $\rho$, Player~$0$ can still enforce value~$b$ from $\p'$. Then, by the definition of adaptive strategies, $\sigma_\rho$ enforces the value $b$ from $\p'$ (in particular for the $(\sigma_\rho,\p')$-play~$\rho$), i.e., we have again derived a contradiction to $\valuation(\rho) < b$.
\end{proof}

\section{Proof of Lemma~\ref{lem:strongly ad if and only if e(p) is summaryMax(v)}}
We use the following properties to simplify the proof.
\begin{remark}
\label{remark:stronglyadaptivestrat}
\hfill
\begin{enumerate}
\item \label{remark:stronglyadaptivestrat:charac} Let $\sigma$ be an adaptive strategy for Player~$0$. Then, $\sigma$ is strongly adaptive if and only if $\stratsummary{\sigma}{\p} = \summaryMax(\p)$  for every play prefix $\p$.

\item \label{remark:stronglyadaptivestrat:simplific} Let $\sigma$ be an adaptive strategy for Player~$0$  and $\p$ a play prefix. Then, $\stratsummary{\sigma}{\p} = \summaryMax(\p)$ if and only if $\stratsummary{\sigma}{\p} \gelex \summaryMax(\p)$.

\item \label{remark:stronglyadaptivestrat:optimumrealized} Let $\p$ be a play prefix. Then, there is a strategy $\sigma$ for Player~$0$ such that $\stratsummary{\sigma}{\p} = \summaryMax(\p)$.

\end{enumerate}	
\end{remark}

Now, let us prove Lemma~\ref{lem:strongly ad if and only if e(p) is summaryMax(v)}. Recall that we need to prove that a strategy $\sigma$ for Player~$0$ in the game $\game'$ is strongly adaptive if and only if for every play prefix~$\p$ ending in vertex~$v$, it holds that $\stratsummary{\sigma}{\p}) = \summaryMax(v)$.

\begin{proof}
By Remark~\ref{remark:stronglyadaptivestrat}.\ref{remark:stronglyadaptivestrat:charac}, a strategy $\sigma$ for Player~$0$ is strongly adaptive if and only if $\stratsummary{\sigma}{\p}= \summaryMax(\p)$ for every play prefix~$\p$. Hence, it is enough to show that for any two prefixes $\p_1$ and $\p_2$ ending in the same vertex, it holds that $\summaryMax(\p_1) = \summaryMax(\p_2)$.
The result then follows by picking $\p_1 = \p$ and $\p_2 = v$.

Suppose, towards a contradiction, and without loss of generality, $\summaryMax(\p_1) \llex \summaryMax(\p_2)$. Let $\sigma_i$ for $i\in \{1,2\}$ be an adaptive strategy that maximizes the summary from $\p_i$ over all strategies. Then, we define the strategy $\sigma'$ for Player~$0$ defined as
\[
\sigma'(\p) = \begin{cases}
 	\sigma_2(\p_2\p') & \text{ if $\p = \p_1\p'$ for some (possibly empty) $\p'$}\\
 	\sigma_1(\p) &\text{ otherwise.}
 \end{cases}
\]
Applying Lemma~\ref{lem:follow adaptive is adaptive} shows that $\sigma'$ is adaptive. Hence, $\stratsummary{\sigma'}{\p_1} = \stratsummary{\sigma_2}{\p_2}$ in the game $\game'$, as $\sigma'$ behaves after the prefix $\p_1$ like $\sigma_2$ does after the prefix~$\p_2$. 
Thus, it holds that 
\[\stratsummary{\sigma'}{\p_1} = \stratsummary{\sigma_2}{\p_2} = \summaryMax(\p_2) \glex \summaryMax(\p_1) = \stratsummary{\sigma_1}{\p_1},\] which contradicts the maximality of $\sigma_1$. 
\end{proof}

\section{Proof of Lemma~\ref{lem:characterization}}

Recall that we need to prove that in the game $\game'$, for some summary $\smry$ and for some vertex~$v$, we have $v\in \smrExact{\smry}$ if and only if $v \not \in \smrMore{\smry}$ and there is an $\smry$-witness from $v$.

\begin{proof}
Fix some $\smry = (b_0, \ldots, b_k,\blanks)\in \summaries$ throughout the proof.

Suppose a vertex~$v \notin \smrMore{\smry}$ has an $\smry$-witness~$\sigma$. 
We show that $v$ is in $\smrExact{\smry}$.

Observe that $v\not \in \smrMore{\smry}$ implies $\summaryMax(v) \lelex \smry$. Furthermore, since $\sigma$ is enforcing, every $(\sigma,v)$-play containing no bad move of Player~$1$ has value at least $b_0$. Hence, the maximum value Player~$0$ enforces from $v$ is at least $b_0$. Therefore, we obtain
	\begin{equation}\label{eqn:lemChar}
	(b_0,\blank,\blank,\blank,\blank)\leq_{lex} \summaryMax(v) \leq_{lex} \smry.
	\end{equation}
If $ k = 0 $, i.e., $\smry = (b_0, \blanks)$ then we are done. 

So, suppose $k > 0$. 
For every play prefix $\p'$ ending in $v'\in \smrMore{\smry}$, let $\sigma_{\p'}$ be an adaptive strategy such that $\stratsummary{\sigma_{\p'}}{\p'} = \summaryMax(\p') = \summaryMax(v')$ (the second equality follows from Lemma~\ref{lem:strongly ad if and only if e(p) is summaryMax(v)}). 
By Remark~\ref{remark:stronglyadaptivestrat}.\ref{remark:stronglyadaptivestrat:optimumrealized}, such a strategy always exists.

We combine these $\sigma_{\p'}$ into a strategy $\sigma_v$ as follows:
for any play prefix~$\p$, $\sigma_v(\p) = \sigma(\p)$ if $\p$ does not contain any vertex in $\smrMore{\smry}$. 
Otherwise, $\sigma_v(\p) = \sigma_{\p'}(\p)$, where $\p'$ is the minimal prefix of $\p$ containing a vertex of $\smrMore{\smry}$, i.e., no strict prefix of $\p'$ contains a vertex in $\smrMore{\smry}$. 
Note that the sets $\smrExact{\smry}$ and $\smrMore{\smry}$ are disjoint, so this is well-defined.
We call $\sigma_v$ the continuation of $\sigma$ with $(\sigma_{\p'})_{\p'}$.

Note that $\sigma_v$ does not make a bad move (of Player~$0$) in any $(\sigma_v, v)$-play, as $\sigma$ enforces~$b_0$ (the largest value that can be enforced from $v$ due to $v \notin \smrEMore{\smry}$) and every bad move of Player~$1$ leading to $\sigma_v$ simulating an adaptive strategy, which does not make any bad move either.
Hence, $\sigma_v$ is also adaptive by Lemma~\ref{lem:follow adaptive is adaptive}. 

We claim that $\stratsummary{\sigma_v}{v} \gelex \smry$. This equality implies $\summaryMax(v) \gelex \stratsummary{\sigma_v}{v} \gelex \smry$.
Then, we have $\summaryMax(v) = \smry$, as we have already argued $\summaryMax(v) \lelex \smry$.
	
So, let us prove $\stratsummary{\sigma_v}{v} \gelex \smry$.
To this end, we show that $\playsummary{\sigma_v}{v}{\rho} \gelex \smry$ for every $(\sigma_v,v)$-uncovered play~$\rho$.
Fix such a play. 
	
As $k > 0$ and due to $\sigma$ being enabling, there is a $(\sigma_v,v)$-play~$\rho_B$ in which Player~$1$ makes at least one bad move. 
Due to Remark~\ref{remark:leftshiftSummary}, we can pick $\rho_B$ such that $\playsummary{\sigma_v}{v}{\rho_B} = \smry$. 
Hence, if $\rho$ does not contain any bad moves by Player~$1$ (which implies $\playsummary{\sigma_v}{v}{\rho} = (b_0, \blanks)$) then $\rho$ is covered by $\rho_B$.
This contradicts our choice of $\rho$.
Thus, we can assume that $\rho$ contains at least one bad move of Player~$1$.
Now, by definition of $\sigma_v$, the play~$\rho$ is consistent with $\sigma$ up to the first bad move of Player~$1$.
	
As $\sigma$ is enforcing, any play~$\rho' \in \bmfree{\sigma}{v}$ is winning w.r.t.\ the parity condition of $\game^{b_0}$.
Hence, the value of $\rho'$ is at least $b_0$.
Thus, $\rho'$ never visits the vertex set~$\enfExact{b_0'}$ for some truth value $b_0'< b_0$.  
Note that by Equation~(\ref{eqn:lemChar}), $v\in \enfExact{b_0}$ as the maximal value Player~$0$ can enforce is $b_0$. Hence, $\rho$ starts in $\enfExact{b_0}$ and it leaves $\enfExact{b_0}$ only when Player~$1$ makes a bad move, which leads to $\enfEMore{b_0}$.
More precisely, the first bad move of Player~$1$ in $\rho$ is a move from some vertex~$v_p$ in $\pre(\smrExact{\smry^\star})$ for some $\smry^\star \in \summaries$ such that the first entry of $\smry^\star$ is strictly greater than $b_0$.
	
As $\sigma$ is evading,  there are two cases: either $\smry^\star \gelex \leftshift{\smry}$ or $\smry^\star$ is a strict prefix of $\leftshift{\smry}$. 
In the second case, by Remark~\ref{remark:leftshiftSummary}, $\rho$ is covered by the play~$\rho_B$, which again contradicts our choice of $\rho$.
In the first case, after the first bad move, the play reaches a vertex~$v^\star$ in $\smrExact{\smry^\star}$.
Let $\p^\star$ be the prefix of $\rho$ ending in this vertex $v^\star$. Thus, $\rho$ is a $(\sigma_{\p^\star},\p^\star)$-play by construction.
There are two subcases:  
either $\rho$ is an $(\sigma_{\p^\star},\p^\star)$-uncovered play and \[\playsummary{\sigma_{\p^\star}}{\p^\star}{\rho}\gelex \stratsummary{\sigma_{\p^\star}}{\p^\star} = \summaryMax(v^\star)=\smry^\star \gelex \leftshift{\smry}\] or
$\rho$ is a $(\sigma_{\p^\star},\p^\star)$-covered play and $\playsummary{\sigma_{\p^\star}}{\p^\star}{\rho}$ is a strict prefix of $\playsummary{\sigma_{\p^\star}}{\p^\star}{\rho'}$ for some $(\sigma_{\p^\star},\p^\star)$-play~$\rho'$.
	
In the second subcase, due to Remark~\ref{remark:leftshiftSummary}, $\rho$ is also a $(\sigma_v,v)$-covered play such that $\playsummary{\sigma_v}{v}{\rho}$ is a strict prefix of $\playsummary{\sigma_v}{v}{\rho'}$, which again contradicts our choice of $\rho$.
In the first subcase, we obtain $\playsummary{\sigma_v}{v}{\rho} \ge \smry$ by Remark~\ref{remark:leftshiftSummary}, which completes the first direction of the proof.
	
For the other direction, suppose a vertex~$v$ belongs to $\smrExact{\smry}$. 
By definition, $\smrExact{\smry} \cap \smrMore{\smry} = \emptyset$, which implies $v \notin \smrMore{\smry}$.
By Remark~\ref{remark:stronglyadaptivestrat}.\ref{remark:stronglyadaptivestrat:optimumrealized}, there exists a strategy~$\sigma$ such that $\stratsummary{\sigma}{v} = \smry$.

It remains to be shown that there is an $\smry$-witness from $v$.
We actually prove a more general result, which will be useful later on: If for some play prefix~$\p$, there exists a strategy~$\sigma$ such that $\stratsummary{\sigma}{\p} = \smry$, then $\sigma$ is an $\smry$-witness from $\p$, i.e., we show the result for arbitrary play prefixes~$\p$.

So, assume we have a strategy~$\sigma$ such that $\stratsummary{\sigma}{\p}  = \smry$.
Hence, there is an $(\sigma, \p)$-uncovered play~$\rho_m$ with $\playsummary{\sigma}{\p}{\rho_m} = \smry$ and $\playsummary{\sigma}{\p}{\rho} \gelex \smry$ for every $(\sigma, \p)$-uncovered play~$\rho$.
We show that the set~$\bmfree{\sigma}{\p}$ of $(\sigma,\p)$-plays without bad moves of Player~$1$ after $\p$ satisfies the three properties of an $\smry$-witness.

The Enforcing property is satisfied by the fact that $\sigma$ enforces the truth value~$b_0$ from $\p$ (as it enforces the first entry of $\smry$ due to $\stratsummary{\sigma}{\p}= \smry$), which implies that the parity condition of $\game^{b_0}$ is satisfied.

Now, assume $k>0$. Then, $\rho_m$ must contain a bad move by Player~$1$. 
Note that the prefix~$\p'$ of $\rho_m$ ending at the position before the first bad move ends in $\pre\big( \smrExact{\leftshift{\smry}}\big)$.
So, there is also a play in $\bmfree{\sigma}{\p}$ that visits this set, but does not contain any bad move by Player~$1$, i.e., an extension of $\p'$ where Player~$0$ uses $\sigma$ and Player~$1$ uses an adaptive strategy for her (which does not make any bad moves). 
Thus, $\sigma$ satisfies the Enabling property.
	
To conclude, assume towards a contradiction, that there is a play in $\bmfree{\sigma}{\p}$ that visits $\pre\big(\smrExact{\smrydash}\big)$ for some summary $\smrydash = (b_0',\ldots, b_{k'}',\blanks) \in \evade{\smry}$, i.e., $\sigma$ does not satisfy the Evading property.
Then, there is a play prefix~$\p'$ extending $\p$ and consistent with $\sigma$ that ends in a vertex~$v'$ such that $\summaryMax(\p') = \summaryMax(v') = \smrydash$ (recall Lemma~\ref{lem:strongly ad if and only if e(p) is summaryMax(v)}).
	
Then, $\stratsummary{\sigma}{\p'} \lelex \smrydash$. 
Hence, by Remark~\ref{remark:existenceofplaysummary}, there exists a $(\sigma, \p')$-uncovered play~$\rho'$ with $\playsummary{\sigma}{\p'}{\rho'} \lelex \smrydash$.
Note that $\rho'$ is also a $(\sigma, \p)$-play, as $\p'$ is a $(\sigma,\p)$-play prefix by construction.
Due to Remark~\ref{remark:leftshiftSummary}, $\playsummary{\sigma}{\p}{\rho'} \lelex (b_0, b_0',\ldots,  b_{k'}',\blanks ) \llex \smry$. 
This is a contradiction to $\playsummary{\sigma}{\p}{\rho} \gelex \smry$ for every $(\sigma,\p)$-uncovered play~$\rho$.
\end{proof}

\section{Proof of Lemma~\ref{lemma:algoforsummary}}
The following remark shows how the maximal summary of a vertex is related to its successor. We use this remark in the next proofs.
\begin{remark}\label{remark:successor}
Let $v \in \smrExact{\smry}$ for $\smry \llex (1111,\blanks)$.
If $v \in V_0'$ then
\begin{itemize}
	\item $v$ has no successor in $\smrMore{\smrydash}$.
	\item $v$ has at least one successor in $\smrExact{\smry}$,
\end{itemize}
If $v \in V_1'$ then:
\begin{itemize}
	\item If $v$ has a successor in $\smrExact{\smrydash}$ for some $\smrydash\llex \smry$ then $\smrydash$ is a strict prefix of~$\smry$.
	\item If $v$ has a successor in $\smrExact{\smrydash}$ for some $\smrydash\glex \smry$ then neither of $\smry$ and $\leftshift{\smry}$ is a strict prefix of~$\smrydash$.
	\item If $v$ does not have a successor in $\smrExact{\smry}$, then it has at least one successor in $\smrExact{\leftshift{\smry}}$.
\end{itemize}
\end{remark}

Now, we give the proof for Lemma~\ref{lemma:algoforsummary}. Recall that we need to prove the our algorithm correctly computes the sets~$\smrExact{\smry}$.

\begin{proof}
As argued earlier, the claim is trivially true for the largest summary, as we have $\smrExact{(1111,\blanks)} = \enfExact{1111}$.

Assuming we already have computed $\smrExact{\smrydash}$ for every $\smrydash \glex \smry$, suppose $U$ is the vertex set computed by the algorithm in that situation. 
We claim $U = \smrExact{\smry}$.

First, we show $U\subseteq \smrExact{\smry}$ by showing that each vertex~$v\in U$ satisfies the characterization given in Lemma~\ref{lem:characterization}. Suppose $v$ is a vertex in $U$. Excluding $\smrMore{\smry}$ in Step~\ref{item:algo1:3} ensures that $v\not \in \smrMore{\smry}$, the first part of the characterization.

If $k = 0$, then since $v\in \enfMore{b_0}$, any adaptive strategy $\sigma$ for Player~$0$ enforces $b_0$ from $v$. Hence, $\sigma$ is enforcing. 
Thus, $v\in \smrExact{\smry}$, as the other two properties of an $\smry$-witness are trivially satisfied in this case.

If $k > 0$, then by Step~\ref{item:algo1:4}, there is a path $v_0 v_1\cdots v_k$ from $v=v_0$ to some $v_k \in \pre\big(\smrExact{\leftshift{\smry}}\big)$ in the arena restricted to $\win{{\smry}}$. 
Let $\sigma({\smry})$ be the winning strategy for Player~$0$ computed in Step~\ref{item:algo1:4} and let $\sigma_0$ be an adaptive strategy for Player~$0$ in the game $\game'$. For $v'\in \smrMore{\smry}$, let $\sigma_{v'}$ be the adaptive strategy that maximizes the summary of $v'$ over all adaptive strategies.  Now consider the strategy $\sigma$ such that
\begin{equation*}
\sigma(\p) =
\begin{cases}
	v_{i+1} &\text{if } \p = v_0v_1\cdots v_i \text{ ending in a Player~$0$ vertex},\\
	\sigma({\smry})(\p) & \text{$\p$ is a play prefix in }\win{{\smry}} \text{ but not a prefix of $v_0 v_1\cdots v_k$},\\
	\sigma_{v'}(v'\p'') &\text{if } \p = \p'v'\p'' \text{ for some play prefix~$\p'$ in }\win{{\smry}},\\
	&\qquad\quad\text{$v'\in \smrMore{\smry}$, and play prefix~$\p''$},\\
	\sigma_0(\p) & \text{otherwise.}
\end{cases}
\end{equation*}
Note that $\sigma$ never makes a bad move and eventually follows an adaptive strategy by construction. Hence, it is also adaptive by Lemma~\ref{lem:follow adaptive is adaptive}.
Let $\rho$ be a  $(\sigma, \p)$-play from a play prefix~$\p$ in $\win{\smry}$. If $\rho$ stays in $\win{{\smry}}\subseteq \enfExact{b_0}$, then it eventually follows $\sigma({\smry})$ and has value at least $b_0$. If not, it follows some strategy $\sigma_{v'}$ for some $v'\in \smrMore{\smry}$ which has value at least $b_0$ by construction. Hence, $\sigma$ satisfies the Enforcing property. 

Moreover, the exists a $(\sigma,v)$-play~$v_0 v_1\cdots v_k \cdots$ (formed by Player~$0$ using $\sigma$ after $v_0 \cdots v_k$ and Player~$1$ using an adaptive strategy after the prefix) showing that $\sigma$ satisfies the Enabling property.

 Furthermore, by removing $\bigcup_{\smrydash\in \evade{\smry}} \reach{\smrExact{\smrydash}}$ in Step~\ref{item:algo1:3}, it is also ensured that $\sigma$ satisfies the Evading property. Hence, $\sigma$ is an $\smry$-witness from $v$, which implies $v\in \smrExact{\smry}$.

For the other direction, we show that $\smrExact{\smry} \subseteq U$ by showing that if a vertex~$v$ satisfies the characterization given in Lemma~\ref{lem:characterization}, then $v\in U$. 

Note that since $U\subseteq \smrExact{\smry}$, the subgraph $\arena_{\smry}$ does not have terminal vertices by Remark~\ref{remark:successor}.
Also, we have $v\not \in \smrMore{\smry}$ by the first item of the characterization. Furthermore, there is an $\smry$-witness~$\sigma$ from $v$.

Hence, by the Enforcing property, $\sigma$ enforces $b_0$ from $v$. Hence, $v\in \enfMore{b_0}$. If $k =0$, then $U = \enfMore{b_0}\setminus \smrMore{\smry}$. Thus, $v\in U$ as required. 

Now suppose $k > 0$. 
If $v\in \reach{\smrExact{\smrydash}}$ for some $\smrydash\in \evade{\smry}$, then Player~$1$ has a strategy $\tau$ that forces the token to reach $\smrExact{\smrydash}$ from $v$ while only visits vertices in $\enfExact{b_0}$. 
Hence, there is a $(\sigma,v)$-play $\rho$ that visits $\pre(\smrExact{\smrydash})$ with $\smrydash\in \evade{\smry}$. 
Hence, there also exists a $(\sigma,v)$-play containing no bad move of Player~$1$ that visiting $\pre(\smrExact{\smrydash})$.
This contradicts $\sigma$ being evading. Hence, $v\not \in \bigcup_{\smrydash\in \evade{\smry}} \reach{\smrExact{\smrydash}}$. Therefore, we obtain $v\in \smrExact{\smry}$. 

Now, as $\sigma$ is enabling, there exists a $(\sigma,v)$-play $\rho$ containing no bad move of Player~$1$ that visits $\pre(\smrExact{\leftshift{\smry}})$.
As $\sigma$ is evading, $\rho$ does not visit $\reach{\smrExact{\smrydash}}$ for any $\smrydash\in \evade{\smry}$. 
 Furthermore, by $\sigma$ being enforcing, since $\rho$ does not contain any bad move of Player~$1$, it stays in the vertex set $\win{{\smry}}$. Hence, $\pre(\smrExact{\leftshift{\smry}})$ is reachable from $v$ in the graph restricted to $\win{{\smry}}$. Therefore, $v\in U$.
\end{proof}

\section{Proof of Lemma~\ref{lem:characterizationStronglyad}}

Recall that we need to prove that in the game $\game'$, a strategy~$\sigma$ is strongly adaptive if and only if it is a $\summaryMax(\p)$-witness from every play prefix $\p$.

\begin{proof}
First, consider a strategy~$\sigma$ that is a $\summaryMax(\p)$-witness from every play prefix $\p$.
Note that $\sigma$ is adaptive by the Enforcing property.
Now we show by induction over summaries~$\smry \in \summaries$, from largest to smallest, that $\stratsummary{\sigma}{\p} \gelex \smry$ for every play prefix~$\p$ ending in $\smrExact{\smry}$, which implies that $\sigma$ is strongly adaptive.

So, for the induction start, we have to consider~$\smry = (1111,\blanks)$, the maximal summary.
So, let $\rho$ be a $(\sigma, \p)$-play such that $\p$ ends in $\smrExact{\smry}$.
Note that $\rho$ cannot contain a bad move after the prefix~$\p$, as $1111$ is the maximal truth value, i.e., $\rho \in \bmfree{\sigma}{\p}$.
Hence, $\rho$ satisfies the parity condition of $\game^{1111}$ due to the Enforcing property, i.e., $\rho$ has value~$1111$.
Thus, we have $\stratsummary{\sigma}{\p} = \smry$ as required.

For the induction step, we consider some~$\smry \llex (1111, \blanks)$.
The induction hypothesis yields that $\stratsummary{\sigma}{\p} \gelex \smry^\star$ for every play prefix~$\p$ ending in $\smrExact{\smry^\star}$ for some $\smry^\star \glex \smry$.
Hence, we have $\stratsummary{\sigma}{\p'} = \summaryMax(\p') = \summaryMax(v')$ (the last equality follows from Lemma~\ref{lem:strongly ad if and only if e(p) is summaryMax(v)}) for every play prefix $\p'$, where $v'\in \smrMore{\smry}$ is the last vertex of $\p'$.

Now, let $\sigma_{\p'} = \sigma$ for every play prefix~$\p'$ ending in $\smrEMore{\smry}$, and let $\sigma'$ be the continuation of $\sigma$ with $(\sigma_{\p'})_{\p'}$. 
Then, we have $\summaryMax(\sigma',\p) \gelex \smry$ using the same reasoning as in the proof of Lemma~\ref{lem:characterization} (note that $\sigma$ satisfies exactly the properties required for that reasoning).  
The desired result follows by noticing that $\sigma'$ is equal to $\sigma$.

For the other direction, let $\sigma$ be a strongly adaptive strategy.
We need to show that $\sigma$ is a $\summaryMax(\p)$ witness from every play prefix~$\p$.

Due to $\sigma$ being strongly adaptive, we have $\stratsummary{\sigma}{\p} = \summaryMax(\p)$ for every play prefix~$\p$.
Hence, by the argument presented in the second direction of the proof of Lemma~\ref{lem:characterization}, we conclude that $\sigma$ is indeed a $\summaryMax(\p)$-witness from every $\p$.
\end{proof}

\section{Proof Lemma~\ref{lemma:gracioustostronglyad}}
We use the following remarks in the proof.

\begin{remark}\label{remark:qnewWinning}
Any play in $\game_{\smry}$ containing $\qnew$ is both $\strong_{\smry}$ and $\weak_{\smry}$-winning.
\end{remark}

\begin{remark}
Let $\smry\in \summaries$ and $v$ be a terminal vertex in the game $\game'$ restricted $\smrExact{\smry}$. Then, $v$ satisfies the following:
\begin{itemize}
\item $v\in V_1'$.
\item $v$ has no successor in $\smrExact{\smry}$.
\item $v$ has a successor in $\smrExact{\leftshift{\smry}}$.
\end{itemize}
\end{remark}

\begin{remark}\label{remark:stratsuccessor}
Suppose $\sigma$ is a strategy of Player~$0$ in $\game'$ that never makes a move from a play prefix ending in $\smrExact{\smry}$ to a vertex in $\smrLess{\smry}$ for some $\smry\in \summaries$. 
Let $\p$ be a play prefix in $\game'$  ending in $\smrExact{\smry}$ and $\rho$ be a $(\sigma,\p)$-play in $\game'$.
If $\rho$ visits a vertex $v'\in \smrExact{\smrydash}$ after~$\p$ for some $\smrydash\in \summaries$, then one the following holds:
		\begin{itemize}
		\item $\smry = \smrydash$.
		\item $\smrydash\glex\smry$ such that $\smry$ is not a strict prefix of $\smrydash$.
		\item $\smrydash$ is a strict prefix of $\smry$.
		\end{itemize}
Moreover, since there are only finitely many summaries and $\rho$ can not return to the same set $\smrExact{\smrydash}$ after leaving the set, it holds that $\rho$ has a suffix staying in $\smrExact{\smrydash}$ forever for some $\smrydash\in \summaries$ satisfying one of the above.
\end{remark}

Let us now prove the lemma.
Recall that we need to show that there exists a strongly adaptive strategy in $\game'$ if and only if there exists a uniformly gracious strategy in every obliging game $\game_{\smry}$.
Moreover, given a uniformly gracious strategy with finite memory in each obliging game $\game_{\smry}$, one can effectively combine these into a strongly adaptive strategy with finite memory in $\game'$.

\begin{proof}
First, assume there exists a strongly adaptive strategy $\sigma$ in the game~$\game'$. Let $\smry = (b_0,\ldots,b_k,\blanks)\in \summaries$. We show that $\sigma$ is a uniformly gracious strategy in the obliging game~$\game_{\smry}$.

Let $\rho$ be a $(\sigma,\p)$-play in $\game_{\smry}$ for some play prefix $\p$ ending in $\smrExact{\smry}$. 
We show that $\rho$ satisfies the strong winning condition and that there is a $(\sigma,\p)$-play~$\rho'$ that satisfies the weak condition.
This implies that $\sigma$ is uniformly gracious.

If $\rho$ contains~$\qnew$, then it is both $\strong_{\smry}$ and $\weak_{\smry}$-winning by Remark~\ref{remark:qnewWinning}, i.e., we can use $\rho' = \rho$ to finish the argument.

So, now suppose $\rho$ does not contain the vertex $\qnew$. Then, $\rho$ is also a $(\sigma,\p)$-play in the game $\game'$.  Since $\sigma$ is enforcing, $\rho$ satisfies the parity condition of $\game^{b_0}$, which implies it also $\strong_{\smry}$-winning.
If $k = 0$, then the weak condition is satisfied by every play, i.e., we  can again use $\rho' = \rho $ to finish the argument.

Otherwise, i.e., if $k > 0$, by the Enabling property, there exists a $(\sigma,\p)$-play~$\rho'$ in $\game'$ that visits $\pre(\leftshift{\smry})$. Let $\p'$ be the minimal prefix of $\rho$ containing a vertex in  $\pre(\leftshift{\smry})$ after $\p$. We claim that all vertices between $\p$ to $\p'$ are in $\smrExact{\smry}$. This implies that $\p'$ is also a $(\sigma,\p)$-play prefix in $\game_{\smry}$. 
Then, by iterating this argument ad infinitum, we obtain a $(\sigma,\p)$-play in $\game_{\smry}$ that visits $\pre(\leftshift{\smry})$ infinitely often. This play satisfies the weak condition.

Now, we only need to prove that all vertices between $\p$ to $\p'$ are in $\smrExact{\smry}$. Let $v$ be a vertex in $\p'$ after $\p$. First, note that a strongly adaptive strategy does not make a move from a play prefix $\p_1$ ending in $\smrExact{\smry_1}$ to a vertex $v_1$ in $\smrLess{\smry_1}$. 
If it would then every $(\sigma,\p_1)$-uncovered play is also a $(\sigma,\p_1 v_1)$-uncovered play and vice versa.
This implies that $\stratsummary{\sigma}{\p_1}$ and $\stratsummary{\sigma}{\p_1 v_1}$ are equal, as they are obtained by minimizing over the same set of plays.
This contradicts the assumption that $v_1\in \smrLess{\smry}$. Hence, by Remark~\ref{remark:stratsuccessor}, $v\in \smrExact{\smrydash}$ for some $\smrydash\in \summaries$ such that $\smrydash\gelex \smry$ or $\smrydash$ is a strict prefix of $\smry$. 

Due to Remark~\ref{remark:existenceofplaysummary}, $\p'$ can be extended to a $(\sigma,\p)$-uncovered play $\rho^\star$ satisfying $\playsummary{\sigma}{\p}{\rho^\star} = \smry$.  Since Player~$1$ makes $k > 0$ bad moves after $\p'$ in $\rho^\star$ (as $\p'$ is the minimal extension of $\p$ visiting a vertex where Player~$1$ can make a bad move), we also have $\playsummary{\sigma}{\p_v}{\rho^\star} = \smry$ for the prefix $\p_v$ of $\p'$ ending in $v$. 
Recall that either $\smrydash\gelex \smry$ or $\smrydash$ is a strict prefix of $\smry$. We show that both $\smrydash\glex \smry$ and $\smrydash$ being a strict prefix of $\smry$ lead to a contradiction, leaving us only with the conclusion  $\smrydash = \smry$ as required.

First, assume we have $\smrydash\glex \smry$. 
Since $\stratsummary{\sigma}{\p_v} = \smrydash$, every $(\sigma, \p_v)$-uncovered play~$\rho''$ satisfies $\playsummary{\sigma}{\p_v}{\rho''} \gelex \smrydash$. But the $(\sigma, \p_v)$-play $\rho^\star$ has summary~$\smry \llex \smrydash$.
We show that $\rho^\star$ is $(\sigma, \p_v)$-uncovered, which yields the desired contradiction.
If $\rho^\star$ is $(\sigma,\p_v)$-covered by another $(\sigma,\p_v)$-play $\rho''$, then $\rho^\star$ is also $(\sigma,\p)$-covered by $\rho''$, which contradicts the assumption that $\rho^\star$ is $(\sigma, \p)$-uncovered.

Finally, assume $\smrydash$ is a strict prefix of $\smry$. By definition, $\stratsummary{\sigma}{\p_v} = \smrydash$ implies that there is some $(\sigma, \p_v)$-uncovered play~$\rho''$ with $\playsummary{\sigma}{\p_v}{\rho''} =  \smrydash$. However, $\rho^\star$ is also a $(\sigma, \p_v)$-play and it covers $\rho''$, as $\smrydash$, the summary of $\rho''$, is a strict prefix of $\smry$, the  summary of $\rho^\star$. 

For the other direction, assume there is a uniformly gracious strategy~$\sigma_{\smry}$ for every game $\game_{\smry}$. 
Let $\sigma$ be the strategy obtained by combining all strategies~$\sigma_{\smry}$ as follows: for any play prefix~$\p$ ending in $\smrExact{\smry}$, we have $\sigma(\p) = \sigma_{\smry}(\p')$, where $\p'$ is the longest suffix of $\p$ which is a $\sigma_{\smry}$-play in $\game_{\smry}$. 
It remains to show that, for every play prefix~$\p$ ending in $\smrExact{\smry}$ for some $\smry = (b_0,\ldots,b_k,\blanks) \in \summaries$, $\sigma$ is an $\smry$-witness from $\p$. This implies that $\sigma$ is strongly adaptive by Lemma~\ref{lem:characterizationStronglyad}.

First, we show that $\sigma$ satisfies the Enabling property. 
If $k = 0$ then the property is satisfied trivially. 
If not, then let $\p'$ be the longest suffix of $\p$ which is a $\sigma_{\smry}$-play in $\game_{\smry}$. 
Then, every $(\sigma,\p)$-play that stays in $\smrExact{\smry}$ follows $\sigma_{\smry}$, i.e., it has a suffix that is a $(\sigma_{\smry},\p')$-play. 
Since, $\sigma_{\smry}$ is uniformly gracious, there exists a  $(\sigma_{\smry},\p')$-play $\rho^\star$ in $\game_{\smry}$ that is $\weak_{\smry}$-winning, i.e., $\rho^\star$ visits $\pre(\leftshift{\smry})\cup \{\qnew\}$ infinitely often. 
Note that by Remark~\ref{remark:qnewWinning}, every predecessor of $\qnew$ also belongs to $\pre(\leftshift{\smry})$. 
Hence, $\rho^\star$ visits $\pre(\leftshift{\smry})$ at least once. 
Let $\p'\p^\star$ be a prefix of $\rho^\star$ ending in $\pre(\leftshift{\smry})$. 
Then, there exists a $(\sigma,\p\p^\star)$-play which is also a $(\sigma,\p)$-play satisfying the Enabling property, i.e., one in which Player~$1$ does not make a bad move. 

Now, given a play $\rho\in \bmfree{\sigma}{\p}$ containing no bad move of Player~$1$ after $\p$, we show that $\rho$ satisfies the Enforcing and the Evading property.

Note that by construction, $\sigma$ never makes a move from a play prefix ending in $\smrExact{\smry}$ to a vertex in $\smrLess{\smry}$. Hence, by Remark~\ref{remark:stratsuccessor}, $\rho$ has a suffix that stays in $\smrExact{\smrydash}$ forever for some $\smrydash = (b_0',\ldots,b_{k'}',\blanks)\in \summaries$ such that $\smrydash\gelex\smry$ or $\smrydash$ is a strict prefix of $\smry$. In any case, $b_0'\geq b_0$. Furthermore, $\rho$ has a suffix that is a $\sigma_{\smrydash}$-play which is $\strong_{\smrydash}$-winning. Hence, it has value at least $b_0'$. Therefore, $\rho$ satisfies the Enforcing property.

Now, suppose $\rho$ does not satisfy the Evading property and visits some vertex $v\in \pre(\smrExact{\smry^\star})$ for some $\smry^\star=(b_0^\star,\ldots,b_{k^\star}^\star,\blanks)\in \evade{\smry}$.
Suppose $v\in \smrExact{\smrydash}$ for $\smrydash = (b_0',\ldots,b_{k'}',\blanks)\in \summaries$.
We claim that $\smry^\star \in \evade{\smrydash}$. 

Step~\ref{item:algo1:3} of the algorithm to compute $\smrExact{\smry}$ and its proof of correctness imply that $\reach{V_{\smry^\star}}$ and $V_{\smrydash}$ are disjoint. 
Hence, the facts that $v$ is a vertex of Player~$1$ (since there exists a bad move from $v$ to $\smrExact{\smry^\star}$) and $v \in \pre(\smrExact{\smry^\star})$, which implies $v \in \reach{\smrExact{\smry^\star}}$, yield the desired contradiction. 

Now, we only need to prove $\smry^\star \in \evade{\smrydash}$ to conclude the proof. First, note that since $\rho$ does not contain a bad move of Player~$1$ after $\p$, we have $b_0' = b_0$, which implies $b_0^\star>b_0'$. To prove the other two conditions, we consider two cases derived as follows: By Remark~\ref{remark:stratsuccessor}, it holds that $\smrydash\gelex\smry$ or $\smrydash$ is a strict prefix of $\smry$.

If $\smrydash\gelex\smry$, then $\smry^\star \llex \leftshift{\smry} \lelex \leftshift{\smrydash}$ (as $b_0 = b_0'$). Furthermore, if $\smry^\star$ is a strict prefix of $\leftshift{\smrydash}\gelex\leftshift{\smry}$, then either $\smry^\star \gelex \leftshift{\smry}$ or $\smry^\star$ is a strict prefix of $\leftshift{\smry}$. This contradicts the fact that $\smry^\star \llex \leftshift{\smry}$ and $\smry^\star$ not being a strict prefix of $\leftshift{\smry}$.

If $\smrydash$ is a strict prefix of $\smry$, then $\smry^\star$ is not a strict prefix of $\smrydash$ (as it is not a strict prefix of $\smry$). Furthermore, since $\smry^\star \llex \leftshift{\smry}$, either $\smry^\star \llex \leftshift{\smrydash}$ or $\leftshift{\smrydash}$ is a strict prefix of $\smry^\star$. 
The second case never holds by Remark~\ref{remark:successor}.
Therefore, the claim is proved.
\end{proof}

\end{document}